\documentclass[a4paper]{paper}
\usepackage{graphicx}
\usepackage{fullpage}
\usepackage{amsmath,amssymb,amsthm}
\graphicspath{{./figs/}}

\newtheorem{theorem}{Theorem}[section]
\newtheorem{lemma}[theorem]{Lemma}
\newtheorem{corollary}[theorem]{Corollary}
\newtheorem{definition}[theorem]{Definition}
\usepackage[utf8x]{inputenc} 
\usepackage[T1]{fontenc}
\usepackage{microtype}       
\usepackage[english]{babel}  
\usepackage{graphicx}
\usepackage[usenames,dvipsnames]{xcolor}
\usepackage{enumitem} 
\usepackage[noadjust]{cite}
\usepackage{amssymb,amsfonts}
\usepackage{thmtools}
\usepackage{thm-restate}
\usepackage{mathrsfs}
\usepackage[ruled,noend,linesnumbered,algosection]{algorithm2e}
\newenvironment{alg}{
  \begin{algorithm}[htb]
    \SetKwInput{In}{input}
    \SetKwInput{Out}{output}
  }{\end{algorithm}}
\usepackage{xspace}
\newcommand{\abbrev}[2]{\expandafter\newcommand\csname #1\endcsname{#2\xspace}}
\abbrev{Caratheodory}{Cara\-th\'eo\-dory}
\abbrev{Barany}{B\'ar\'any}
\abbrev{etal}{et al.\@}
\abbrev{WordRAM}{\textsc{Word-Ram}}
\abbrev{RealRAM}{\textsc{Real-Ram}}

\newenvironment{prf}{\begin{proof}}{\end{proof}}

\newcommand{\paranthesis}[1]{\left(#1\right)}
\newcommand{\s}{\star}
\newcommand{\lt}{\left}
\newcommand{\rt}{\right}

\DeclareFontFamily{U}{mathx}{\hyphenchar\font45}
\DeclareFontShape{U}{mathx}{m}{n}{
      <5> <6> <7> <8> <9> <10>
      <10.95> <12> <14.4> <17.28> <20.74> <24.88>
      mathx10
      }{}
\DeclareSymbolFont{mathx}{U}{mathx}{m}{n}
\DeclareFontSubstitution{U}{mathx}{m}{n}
\DeclareMathAccent{\widecheck}{0}{mathx}{"71}
\DeclareMathAccent{\wideparen}{0}{mathx}{"75}

\newcommand{\DOWN}[1]{{\widecheck{#1}}}
\newcommand{\up}[1]{{\hat{#1}}}
\newcommand{\UP}[1]{{\widehat{#1}}}

\newcommand{\set}[1]{\left\{#1\right\}}
\newcommand{\midd}{\,\middle\vert\,}
\newcommand{\eps}{\ensuremath{\varepsilon}}

\DeclareMathOperator{\poly}{poly}
\DeclareMathOperator{\aff}{aff}
\DeclareMathOperator*{\argmin}{arg\,min}

\DeclareMathOperator{\conv}{conv}
\DeclareMathOperator{\lsp}{span}
\DeclareMathOperator{\pos}{pos}

\newcommand{\convv}[1]{\conv\paranthesis{#1}}
\newcommand{\poss}[1]{\pos\paranthesis{#1}}
\newcommand{\lspp}[1]{\lsp\paranthesis{#1}}

\newcommand{\N}{\ensuremath{\mathbb{N}}}
\newcommand{\Q}{\ensuremath{\mathbb{Q}}}
\newcommand{\R}{\ensuremath{\mathbb{R}}}
\newcommand{\Rp}{\ensuremath{\mathbb{R}_+}}
\newcommand{\Z}{\ensuremath{\mathbb{Z}}}

\newcommand{\ve}[1]{{\boldsymbol #1}}
\newcommand{\0}{\ve{0}}
\newcommand{\1}{\ve{1}}
\newcommand{\e}{\ve{e}}
\newcommand{\bb}{\ve{b}}
\newcommand{\cc}{\ve{c}}
\newcommand{\hh}{\ve{h}}
\newcommand{\pp}{\ve{p}}
\newcommand{\qq}{\ve{q}}
\newcommand{\rr}{\ve{r}}
\newcommand{\xx}{\ve{x}}

\newcommand{\TwoRowVec}[2]{\begin{pmatrix}#1\\ #2\end{pmatrix}}
\newcommand{\Landau}[2]{\expandafter\newcommand\csname
#1\endcsname[1]{\mathbin{\operatorname{#2}}{\left(##1\right)}}}
\Landau{Oh}{O}
\Landau{Th}{\Theta}
\Landau{Om}{\Omega}

\newcommand{\cclasss}[2]{\abbrev{#1}{\textsf{#2}}}
\newcommand{\cclass}[1]{\cclasss{#1}{#1}}
\cclass{NP}
\cclass{coNP}
\cclass{TFNP}
\cclass{PLS}
\cclass{PPAD}
\cclass{PSPACE}

\newcommand{\prob}[1]{\textsc{#1}}
\newcommand{\defprob}[2]{\abbrev{#1}{\prob{#2}}}
\defprob{NCP}{Ncp}
\defprob{ThreeSat}{3Sat}
\defprob{NCPl}{L-Ncp}
\defprob{NCPg}{G-Ncp}
\defprob{MTSATl}{Max-2SAT/Flip}
\defprob{Centerpoint}{Cen\-ter\-point}
\defprob{ColKirchberger}{Col\-or\-ful\-Kirch\-berger}
\defprob{Tverberg}{Tver\-berg}
\defprob{CCP}{Col\-or\-ful\Caratheodory}

\newcommand{\mc}[1]{\ensuremath{\mathcal{#1}}}

\newcommand{\eq}[1]{(\ref{#1})}

\newcommand{\inst}{\mc{I}}
\newcommand{\sol}{\mc{S}}

\newcommand{\cost}{c}
\newcommand{\nbr}{N}

\title{Computational Aspects of the Colorful \Caratheodory Theorem\footnote{WM 
    was supported in part by DFG Grants MU 3501/1 and MU 3501/2 
   and ERC StG 757609. YS was
    supported by the Deutsche Forschungsgemeinschaft within the
    research training group `Methods for Discrete Structures' (GRK
    1408) and by GIF grant 1161. \\
    A preliminary version of this
    article appeared as W.~Mulzer and Y.~Stein.
   \emph{Computational Aspects of the Colorful \Caratheodory Theorem}.
   Proc.~31st SoCG, 2015.}
}

\author{Wolfgang Mulzer
    \and
      Yannik Stein}

\institution{Institut f\"ur Informatik, Freie Universit\"at Berlin, Germany\\
  \texttt{\{mulzer, yannikstein\}@inf.fu-berlin.de}
  }

\begin{document}
\maketitle

\begin{abstract}
Let $C_1,\dots,C_{d+1}\subset\R^d$ be $d+1$ point sets, each containing the
origin in its convex hull. We call these sets \emph{color classes}, and we 
call a sequence $p_1, \dots, p_{d+1}$ with $p_i \in C_i$,
for $i = 1, \dots, d+1$, a \emph{colorful choice}.
The colorful \Caratheodory theorem guarantees the existence of a 
\emph{colorful choice} that also contains the origin in its convex hull. 
The computational complexity of finding such a colorful choice (\CCP) is 
unknown. This is particularly interesting in the light of polynomial-time 
reductions from several related problems, such
as computing centerpoints, to \CCP.

We define a novel notion of approximation that is compatible with the
polynomial-time reductions to \CCP: a sequence that contains at most 
$k$ points from each color class is called a \emph{$k$-colorful choice}. 
We present an algorithm that for any fixed $\eps > 0$, outputs an 
$\lceil \eps d\rceil$-colorful choice
containing the origin in its convex hull in polynomial time.

Furthermore, we consider a related problem of \CCP: in the \emph{nearest
colorful polytope} problem (\NCP), we are given sets 
$C_1,\dots,C_n\subset\R^d$ that do not necessarily contain the origin 
in their convex hulls. The goal is to find a colorful choice whose 
convex hull minimizes the distance to the origin.
We show that computing a local optimum for \NCP is \PLS-complete, while
computing a global optimum is \NP-hard.
\end{abstract}

\section{Introduction}
Let $P\subset\R^d$ be a point set. \Caratheodory's well-known
theorem~\cite[Theorem~1.2.3]{Matouvsek2002} states that the containment 
of each point in $\conv(P)$ can be witnessed by a ``small'' subset of $P$. 
Moreover, the standard
proof of this result is constructive and gives a polynomial-time
algorithm if the coefficients of the original convex combination are known. 
In the following, we say that $P$ \emph{embraces} a point $\qq \in\R^d$ or 
is \emph{$\qq$-embracing} if and only if $\qq$ is in the 
convex hull of $P$. Similarly, we say $P$ \emph{ray-embraces} $\qq$ if
and only if $\qq$ is in the cone spanned by $P$.

\begin{theorem}[\Caratheodory's theorem]\label{thm:caratheodory}
  Let $P=\{\pp_1,\dots,\pp_n\} \subset \R^d$ be a set of $n$ points.
\begin{description}
  \item[(Convex version)]
  If $P$ embraces the origin, there is an affinely independent subset 
  $P' \subseteq P$ that embraces the origin.
  \item[(Cone version)] If $P$ ray-embraces a point $\bb \in \R^d$, there
  is a linearly independent subset $P' \subseteq P$ that ray-embraces $\bb$.
\qed
\end{description} 
\end{theorem}

As we will discuss in Section~\ref{sec:prelim}, the standard proof of
Theorem~\ref{thm:caratheodory} is constructive and can be interpreted as a
polynomial-time algorithm. \Barany~\cite{Barany1982} generalized
\Caratheodory's theorem by introducing colors: now, multiple
point sets embrace the origin, and we think of these point sets as 
\emph{color classes}. Then, there is a sequence of points, one from each
color class, that also embraces the origin.
This is called a \emph{colorful choice}. See
Figure~\ref{fig:colcara} for an example.

\begin{figure}[htbp]
  \begin{center}
    \includegraphics{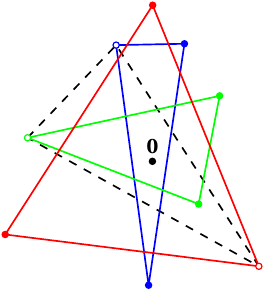}
  \end{center}
  \caption{The colorful \Caratheodory theorem in 
  two dimensions: all color classes embrace the origin 
  and the marked points form a $\0$-embracing colorful choice.}
  \label{fig:colcara}
\end{figure}

\begin{theorem}[Colorful \Caratheodory theorem~\cite{Barany1982}]
  \label{thm:colcara}
  Let $C_1,\dots,C_{d+1} \subset \R^d$ be point sets that all 
  embrace the origin.  There exists a colorful
  choice that embraces the origin.
\end{theorem}
\begin{prf}
  Let $C$, $|C|\leq d+1$, be a colorful
  choice of $C_1,\dots,C_{d+1}$. Let $\Phi(C)$ be the
  minimum $\ell_2$-distance between any point in $\conv(C)$ and the 
  origin.  If $\Phi(C) = 0$, then $\0 \in\conv(C)$, and we 
  are done. Now, assume $\Phi(C) > 0$.
  Let $\cc$ be the point in $\conv(C)$ with minimum
  $\ell_2$-distance to the origin. Furthermore, let $h^-$ be the 
  open halfspace that contains
  the origin and that is bounded by the hyperplane through $\cc$ that is
  orthogonal to $\cc$ interpreted as a vector.
  Since $\cc$ minimizes the distance to the origin, it is contained in a
  facet of $\conv(C)$. Note that $\cc$ is not necessarily contained in the
  interior of a facet.  Theorem~\ref{thm:caratheodory} implies that 
  there is a $d$-subset $F \subset C$ of $C$ with 
  $\cc \in \conv(F)$. Let $i^\times$ be the
  color of the point that is missing in $F$. The halfspace $h^-$ contains 
  the origin, and thus it contains at least one point $\cc_{i^\times} \in
  C_{i^\times}$ with color $i^\times$.
  Now, set $C' = (F \cup \{\cc_{i^\times}\})$. Since
  $\conv(C')$ contains $\cc$ and a point in $h^-$, we have 
  $\Phi(C') < \Phi(C)$. Thus, if $\Phi(C) > 0$, there is
  always a way to strictly decrease it. The situation is depicted in
  Figure~\ref{fig:proof_colcara}.
  Because there is only a finite number of colorful
  choices, there is a colorful choice 
  $C^\s$ with $\Phi(C^\s) = 0$.
\end{prf}

\begin{figure}[htbp]
  \begin{center}
    \includegraphics{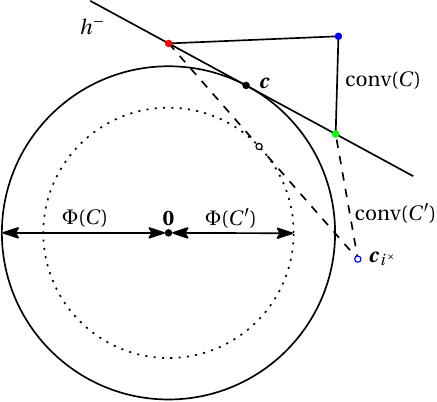}
  \end{center}
  \caption[Proof of the colorful \Caratheodory theorem.]{Proof of the 
  colorful \Caratheodory theorem: if the potential
  function is larger than $0$, it can be decreased by swapping one 
  point with another point of the same color.}
  \label{fig:proof_colcara}
\end{figure}

The convex version of Theorem~\ref{thm:caratheodory} can be derived
directly from Theorem~\ref{thm:colcara} by setting 
$C_1=\dots=C_{d+1}=P$. There are many different variants and
generalizations of the colorful \Caratheodory theorem 
(see~\cite{MeunierDe2013}).

We denote with \CCP the computational problem of finding a $\0$-embracing
colorful choice under the conditions of Theorem~\ref{thm:colcara}.
\CCP is particularly interesting in the light of its
applications: let $P \subset \R^d$ be a point set. We call a partition of 
$P$
into $r$ sets $P_1,\dots,P_r$ a \emph{Tverberg $r$-partition} if the convex
hulls of the $P_i$ have a point in common. By Tverberg's
theorem~\cite{Tverberg1966}, there always exists a Tverberg 
$\lceil |P| / (d+1)\rceil$-partition.
We denote the computational problem of finding such a partition
by \Tverberg. Sarkaria's proof~\cite{Sarkaria1992}
of Tverberg's theorem can be interpreted as polynomial-time
reduction of \Tverberg to \CCP. Moreover, Tverberg's theorem 
directly implies
the centerpoint theorem~\cite{Rado1946} that guarantees the existence of
\emph{centerpoints}, a popular generalization of the 
median to higher dimensions. We call a
point $\qq \in
\R^d$ a \emph{centerpoint} for $P$ if any closed halfspace that 
contains $\qq$ also
contains at least $\lceil |P| / (d+1) \rceil$ points from $P$. Consider a
Tverberg $r$-partition $P_1,\dots,P_r$ of $P$ for $r = \lceil |P| / (d+1)
\rceil$. Then any
point in $\bigcap_{i=1}^r \conv(P_i) \neq \emptyset$ is a centerpoint. 
Hence, the computational problem of computing
centerpoints, \Centerpoint, can again be reduced in polynomial time to 
\CCP. Furthermore, the key argument of Sarkaria's proof of Tverberg's 
theorem can also be used to prove the \emph{colorful Kirchberger 
theorem}~\cite{ArochaBaBrFaMo2009}: given
$n$ Tverberg $r$-partitions $\mc{T}_1, \dots, \mc{T}_n$ for
disjoint $d$-dimensional point sets of size $n$ and $r = \lceil n / (d+1)
\rceil$, a Tverberg $r$-partition $\mc{T}$ can be
constructed by taking exactly one point from each $\mc{T}_i$ and putting 
it in the set of $\mc{T}$ with the same index as in $\mc{T}_i$. Again, 
the proof can be interpreted as a polynomial-time reduction to \CCP from 
\ColKirchberger, 
the computational problem corresponding to the colorful Kirchberger theorem.
We discuss these reductions in more detail in
Section~\ref{sec:kcol:applications}.

In contrast to \Caratheodory's theorem, the complexity of \CCP is
still unsettled. Since a solution always exists and can be verified in
polynomial-time, \CCP is contained in
the complexity class \emph{total function NP} (\TFNP). This already 
implies that
\CCP is not \NP-hard unless $\NP =
\coNP$~\cite[Theorem~2.1]{MegiddoPa1991},~\cite[Lemma~4]{JohnsonPaYa1988}.
In a recent result, Meunier \etal~\cite{MeMuSaSt2017} showed
that \CCP is contained in the intersection of two important subclasses 
of \TFNP:
\emph{polynomial parity argument in a directed graph} (\PPAD) and
\emph{polynomial-time local search} (\PLS).
Moreover, Meunier and Sarrabezolles~\cite{MeunierSa2014} have shown that
a related problem is \PPAD-complete:
given $d+1$ pairs of points $P_1,\dots,P_{d+1}\in \Q^d$ and a 
colorful choice that embraces the origin, find
another colorful choice that embraces the origin.
Complementary to this result, we show in Section~\ref{sec:ncp} 
that a related
problem is \PLS-complete, the \emph{nearest colorful polytope problem} 
(\NCP): given $n$ color classes $C_1,\dots,C_n$, find a
colorful choice whose distance to the origin cannot be decreased 
by swapping one 
point with another point of the same color. This problem is motivated by
\Barany's proof of Theorem~\ref{thm:colcara}. Furthermore, we show that the
global search variant of \NCP is \NP-hard, which answers a question by 
\Barany and Onn~\cite{BaranyOn1997}. This question was also answered 
independently by Meunier and Sarrabezolles~\cite{MeunierSa2014}.

Despite the recent improvements on the upper bounds on the complexity 
of \CCP, a polynomial-time algorithm remains elusive. Hence, 
approximation algorithms are of interest. This was first considered 
by \Barany and Onn~\cite{BaranyOn1997}
who described how to find a colorful choice
whose convex hull is ``close'' to the origin under several general position
assumptions. We call a set \emph{$\eps$-close} to the origin if its convex 
hull
has $\ell_2$-distance at most $\eps$ to $\0$. Let in the following 
$\eps, \rho > 0$ be
parameters. Given $d+1$ sets $C_1,\dots,C_{d+1}\in \Q^d$ such that
\begin{enumerate}[label=(\roman{enumi})]
\item each $C_i$, $i \in [d+1]$, contains a ball of radius $\rho$ 
centered at
the origin in its convex hull, and
\item all points $\pp\in C_i$, $i \in[d+1]$, fulfill 
$1\leq \|\pp\| \leq 2$.
\end{enumerate}
Then, the algorithm by \Barany and Onn iteratively computes a sequence of
colorful choices such that the $\ell_2$-distances of their convex hulls 
to the origin
strictly decrease until a colorful choice that embraces the origin is
found. In particular, if stopped earlier, a colorful choice that is 
$\eps$-close to $\0$
can be computed in time $\text{poly}(L, \log(1/\eps),1/\rho)$ on the
\textsc{Word-Ram} with logarithmic costs. Here, $L$ denotes the
length of the bit-encoding of the input points. Note that if $1/\rho =
\Oh{\text{poly}(L)}$, the algorithm actually finds a colorful choice 
that embraces the origin in polynomial-time. The \Barany-Onn algorithm
is essentially the algorithm from the proof of the convex version of
Theorem~\ref{thm:colcara}, and the main contribution is a careful analysis.

In the same spirit, Barman~\cite{barman2015} showed
that if the points have constant norm, a colorful choice that is 
$\eps$-close to
the origin can be found in $d^{\Oh{1/\eps^2}}L$ time, where $L$ is again
the length of the input encoding. The algorithm uses
the following approximate version of \Caratheodory's theorem as a main
ingredient: let $P \subset \R^d$ be a $\0$-embracing point set.
Then, for any $\eps > 0$, there exists a subset $P' \subseteq P$ of size
$c_{\eps} = \Oh{\max_{\pp \in P} \| \pp \| / \eps^2}$ that is 
$\eps$-close to
$\0$. This immediately implies a simple brute-force
algorithm: let $C_1,\dots, C_{d+1} \subset \Q^d$ be point sets with 
$\0 \in \conv(C_i)$, for $i \in [d+1]$, and assume all points 
have constant norm. Let further $C \subseteq
\bigcup_{i=1}^{d+1} C_i$ be a $\0$-embracing colorful choice whose 
existence is guaranteed by Theorem~\ref{thm:colcara}. Then, the 
approximative version of \Caratheodory's theorem asserts that there 
is a subset $C' \subseteq C$ of size $c_{\eps}$ that is $\eps$-close 
to the origin. We can now guess $C'$ by trying out all
$\binom{d+1}{c_{\eps}}$ possibilities for the colors in $C'$, and 
for each color $i$, we try all $|C_i|$ possibilities
of picking a point with color $i$. For each choice of $C'$, we can
check whether it is $\eps$-close to the origin by solving a convex 
quadratic program. Solving one convex quadratic program needs 
$\Oh{\poly(d)L}$ time~\cite{KapoorVa1986,kozlov1980polynomial}. 
Hence, assuming that each color class
is of size $\Oh{d}$, we can compute an $\eps$-close colorful choice in
$d^{\Oh{1/\eps^2}}L$ time.

It is desirable to approximate \CCP in a way that is compatible with the
polynomial-time reductions to it. Then, good enough approximation algorithms
for \CCP can be converted to approximation algorithms for \Tverberg,
\Centerpoint, and \ColKirchberger. Both approximation algorithms above 
relax the requirement that
the resulting colorful choice embraces the origin. However, in the
polynomial-time reductions from \Tverberg, \Centerpoint, and 
\ColKirchberger to \CCP, it is crucial that the colorful choice embrace 
the origin. If the convex hull is only close to the origin but does not 
contain it, the reductions 
break down, and it is not immediate how to fix them. On
the other hand, allowing multiple points from each color class has a natural
interpretation in the polynomial-time reductions to \CCP and leads to
approximation algorithms for the other problems. Let 
$C_1,\dots,C_{d+1} \subset \R^d$ be point sets that embrace the origin and 
let $k \in \N$ be a number. We call a set 
$C \subseteq \bigcup_{i=1}^{d+1} C_i$ a \emph{$k$-colorful choice} if it
contains at most $k$ points from each $C_i$. In 
Section~\ref{sec:kcol:applications}, we assume an oracle that 
computes $\0$-embracing $k$-colorful
choices, and we give precise bounds on the quality of the
approximation algorithms for
\Tverberg, \Centerpoint, and \ColKirchberger depending on $k$. 
We obtain these bounds by combining this
oracle with the polynomial-time reductions. Furthermore, in
Section~\ref{sec:kcol}, we present an algorithm that computes for any fixed
$\eps > 0$, a $\0$-embracing $\lceil \eps d\rceil$-colorful choice.

\section{Preliminaries: Embracing Equivalent Points}
\label{sec:prelim}
Throughout the paper, vectors or points are set in
boldface. The origin is denoted by $\0$, the canonical basis of $\R^d$ is
denoted by $\e_1,\dots,\e_d$, and the all-ones vector 
$\sum_{i=1}^{d} e_i$ is denoted by $\1$. For a set of points 
$P = \lt\{\pp_1, \dots, \pp_n \rt\} \subset\R^d$, we denote by
\begin{itemize}
  \item $\lspp{P} =
  \lt\{\sum_{i = 1}^n \phi_i \pp_i \midd \phi_i \in \R\rt\}$
  its linear span and the subspace orthogonal to $\lsp(P)$ by
  ${\lspp{P}}^\perp = \lt\{v \in \R^d \midd \forall p \in \lsp(P):
    \langle v, p \rangle = 0\rt\}$;
  \item $\aff(P) = \lt\{\sum_{i = 1}^n \alpha_i \pp_i \midd
    \alpha_i \in \R, \sum_{i = 1}^n \alpha_i = 1 \rt\}$ its affine hull;
  \item $\pos(P) = \lt\{\sum_{i=1}^n \psi_i \pp_i \midd \psi_i \in 
  \R_+ \rt\}$ all linear combinations with nonnegative coefficients. 
  We call $\pos(P)$ the \emph{positive span} of $P$ and we call a 
  combination with nonnegative coefficients a \emph{positive combination};
  \item $\conv(P) = 
  \lt\{\sum_{i = 1}^n \lambda_i \pp_i \midd \lambda_i \in \R_+,
    \sum_{i=1}^n \lambda_i = 1 \rt\}$ its convex hull;
  \item $\dim P$ the dimension of $\lsp(P)$;
\end{itemize}

Unless noted otherwise, all algorithms are analyzed in the
\RealRAM model of 
computation~\cite[Chapter~1.4]{PreparataSh1985}.\footnote{Recall
that the \RealRAM is the standard model of computational geometry
where memory cells store arbitrary real numbers and operations
on them can be performed at unit cost. We emphasize that there is no known 
algorithm for solving linear programs
that needs a polynomial number of steps on the \RealRAM. Thus,
our algorithms avoid the use of LPs.}
We begin with a constructive
version of Theorem~\ref{thm:caratheodory}.

\begin{lemma}[Constructive version of \Caratheodory's theorem]
\label{lem:constr_cara}
Suppose that $P \subset \R^d$ is a $\0$-embracing point set.
Given the coefficients of the convex combination of $\0$ with the points in
$P$, a $\0$-embracing affinely independent subset $P' \subseteq P$ can be
computed in $\Oh{d^3 |P| + |P|^2}$ time.
\end{lemma}
\begin{prf}
The standard proof of Theorem~\ref{thm:caratheodory} is already 
constructive. We repeat it briefly before analyzing its running time 
when interpreted as an algorithm.

Assume $P$ is affinely dependent. Let $\pp_1,\dots,\pp_n$ denote the 
points in $P$ and let
$\alpha_1,\dots,\alpha_n \in \R$ be coefficients of a nontrivial affine
dependency, i.e., let
\begin{equation}\label{eq:cara:1}
\0 = \alpha_1 \pp_1 + \dots + \alpha_n \pp_n
\end{equation}
with $\sum_{i=1}^{n} \alpha_i = 0$ and $\alpha_i > 0$ for some $i \in [n]$.
Furthermore, because $\0 \in \conv(P)$, there are coefficients
$\lambda_1,\dots,\lambda_n \in \R_+$ such that
\begin{equation}\label{eq:cara:2}
\0 = \lambda_1 \pp_1 + \dots + \lambda_n \pp_n
\end{equation}
and $\sum_{i=1}^{n} \lambda_i = 1$.
Let $c \in \R$ be a factor that is to be specified.
Scaling (\ref{eq:cara:1}) by $c \in \R$ and subtracting it from
(\ref{eq:cara:2}), we obtain
\[
 \0 = \sum_{i=1}^{n} \lambda_i \pp_i - c \sum_{i=1}^{n}  \alpha_i \pp_i
  = \sum_{i=1}^{n} \lambda'_i \pp_i,
\]
where $\lambda'_i = \lambda_i - c \alpha_i$.  Thus, let $i^\star =
\argmin\lt\{\lambda_i/\alpha_i
\midd i \in [n], \alpha_i > 0\rt\}$, where ties are broken arbitrarily, 
and set $c = \lambda_{i^\star} / \alpha_{i^\star}$. Then,
$\sum_{i=1}^{n} \lambda'_i \pp_i$ is a convex combination
of $\0$ with the points in $P \setminus \{\pp_{i^\star}\}$.
Indeed by definition of $i^\star$, we have $\lambda'_i = \lambda_i - c
\alpha_i \geq 0$, $\sum_{i=1}^{n} \lambda'_i = \sum_{i=1}^{n} (\lambda_i - c
\alpha_i) = \sum_{i=1}^{n} \lambda_i = 1$, and $\lambda'_{i^\star}= 0$.
A repeated removal of points until the remaining set is affinely independent
implies the statement.

It remains to show the running time. We compute in each
iteration a linear
dependency by Gaussian elimination in $\Oh{d^3}$ time.\footnote{On
the \RealRAM, we need not worry about the bit-complexity
of Gaussian elimination.}
By our assumption, we know the convex coefficients 
$\lambda_1,\dots,\lambda_n$ and thus, we can find the
point $\pp_{i^\star}\in P$ in $\Oh{n}$ time. Furthermore, we can compute the
new coefficients $\lambda'_i \in \Rp$, $i \in [n] \setminus \{i^\star\}$,
from $\lambda_1,\dots,\lambda_n$, the coefficients of the affine 
dependency, and the index $i^\star$ in $\Oh{n}$ time. Hence, one
iteration takes $\Oh{d^3 + n}$ time and since there are $\Oh{n}$ 
iterations, the algorithm needs in total $\Oh{d^3 n + n^2}$ time.
\end{prf}

In Section~\ref{sec:kcol}, we present two approximation algorithms 
that follow the same strategy: begin with a complete color class and then
replace a subset by points from other color classes while maintaining the
property that the origin is embraced. We conclude this section with the
necessary tools to implement the replacement step.

Let $C \subset \R^d$ be a $\0$-embracing point set. We say $C$ is
\emph{minimally $\0$-embracing} if $C \setminus \{\cc\}$ is not 
$\0$-embracing for all points $\cc \in C$.

\begin{figure}[htbp]
  \begin{center}
    \includegraphics{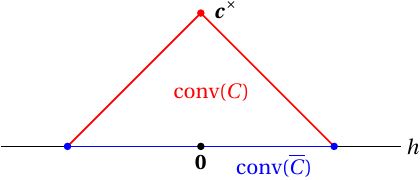}
  \end{center}
  \caption[Minimal $\0$-embracing sets.]{The blue points constitute the 
  linearly dependent set $\overline{C}$.
  The removal of $\cc^\times$ maintains the embrace of the origin.}
  \label{fig:minembr}
\end{figure}

\begin{lemma}\label{lem:lindepembracing}
  Let $C \subset \R^d$ be an affinely independent $\0$-embracing set. 
  Then, a subset $C'$ of $C$ is linearly dependent if and only if 
  $C'$ embraces the origin.
\end{lemma}
\begin{prf}
  Clearly, all $\0$-embracing subsets of $C$ must be linearly
  dependent. Let now $C'$ be a linearly dependent subset of $C$. We need 
  to show that $C'$ is $\0$-embracing. Assume without loss of generality 
  that $C'$ is a proper subset and let $\cc^\times \in C \setminus C'$ be 
  a missing point.  We prove that the set 
  $\overline{C} = C \setminus \{\cc^\times\}$ is $\0$-embracing. 
  A repeated application of this argument then implies the statement.

  Since $C' \subseteq \overline{C}$, the set $\overline{C}$ is linearly
  dependent. Thus, we can write $\0$ as a nontrivial linear
  combination $\sum_{\cc \in \overline{C}} \phi_\cc \cc$ of
  the points in $\overline{C}$, where $\phi_{\cc} \in \R$, for all $\cc \in
  \overline{C}$. Furthermore, since $C$ is affinely independent,
  so is $\overline{C}$, and hence $\sum_{\cc \in \overline{C}} \phi_\cc 
  \neq 0$.
  By rescaling the coefficients, we obtain an affine combination of $\0$.
  This implies $\aff(\overline{C}) = \lsp(\overline{C})$. Now, because
  $\overline{C} = C \setminus \{\cc^\times\}$ and because $C$ is affinely
  independent, the point $\cc^\times$ is not contained in the affine hull of
  $\overline{C}$ and thus not in the linear span of $\overline{C}$. Then, 
  there exists a hyperplane $h$ that contains $\lsp(\overline{C})$ but 
  not $\cc^\times$.  See Figure~\ref{fig:minembr}.
  Because $\conv(C)$ is on one side of $h$, the
  intersection $h \cap \conv(C) = \conv(\overline{C})$ is a face of
  $\conv(C)$. Since $h$ and $\conv(C)$ both contain the origin, the face
  $\conv(\overline{C})$ must contain the origin, too. Hence, 
  $\overline{C}$ is $\0$-embracing.
\end{prf}

\begin{lemma}\label{lem:minembr}
  Let $C \subset \R^d$ be a minimally $\0$-embracing set. Then, the 
  following holds:
  \begin{enumerate}[label=(\roman{enumi})]
    \item\label{lem:minembr:sublinindep} $C$ is affinely independent and all
    proper subsets of $C$ are linearly independent.
    \item\label{lem:minembr:rayembr} For all $\cc \in C$, the point 
    $-\cc$ is ray-embraced by $C \setminus \{\cc\}$.
  \end{enumerate}
  In particular, $\dim C = |C|-1$ and $\pos(C) = \lsp(C)$.
\end{lemma}
\begin{prf}
  If $C$ is affinely dependent, then by Theorem~\ref{thm:caratheodory} 
  there exists
  a proper subset that embraces the origin. Thus, $C$ must be affinely
  independent. Hence,~\ref{lem:minembr:sublinindep} is implied by
  Lemma~\ref{lem:lindepembracing}. Write now $C$ as $\cc_1,\dots,\cc_n$ 
  and let $\lambda_1,\dots,\lambda_n \in \R_+$ be coefficients that sum 
  to $1$ such that $\0 = \sum_{i=1}^{m} \lambda_i \cc_i$. Then, 
  $-\lambda_i \cc_i \in \pos(C)$ for all $i \in [n]$. Because 
  $C \setminus \{\cc\}$ does not
  embrace the origin for any $\cc \in C$, we have 
  $\lambda_i > 0$ for $i \in [n]$. This implies~\ref{lem:minembr:rayembr}.
\end{prf}

Using the fact that all proper subsets of a minimally $\0$-embracing set 
$C$ are linearly independent, we show how to compute for each point in the 
positive span of $C$ the coefficients of the positive combination.

\begin{lemma}\label{lem:minembcoeff}
  Let $C \subset \R^d$ be a minimally $\0$-embracing set and let $\qq \in
  \poss{C}$ be a point. Then, we can compute the
  coefficients of a nontrivial positive combination of $\qq$ with 
  the points in $C$ in $\Oh{d^4}$ time.
\end{lemma}
\begin{prf}
  Consider first the case that $\qq = \0$.
  Let $\cc^\times \in C$ be an arbitrary point and denote
  with $\overline{C} = C \setminus \{\cc^\times\}$ the remaining points.
  By Lemma~\ref{lem:minembr}, $-\cc^\times$ is ray-embraced by 
  $\overline{C}$. Thus, the
  linear system $A \xx = -\cc^\times$, where $A$ is the matrix whose 
  columns are the
  points from $\overline{C}$, has a solution. By
  Lemma~\ref{lem:minembr}~\ref{lem:minembr:sublinindep}, the
  set $\overline{C}$ is linearly independent and hence this solution is 
  unique.
  Thus, we
  can compute the coefficients $\psi_{\cc} \in \R$, $\cc
  \in \overline{C}$, such that $-\cc^\times = \sum_{\cc \in \overline{C}}
  \psi_{\cc} \cc$ in $\Oh{d^3}$ time with Gaussian
  elimination. Moreover, since the solution is unique, we must have
  $\psi_{\cc} \geq 0$ for all $\cc \in \overline{C}$. Set 
  $\psi_{\cc^\times}$ to $1$.
  Then, $\0 = \sum_{\cc \in C} \psi_{\cc} \cc$, all coefficients are
  nonnegative, and not all coefficients are zero.

  Consider now the case that $\qq \neq \0$. We iterate through all 
  $\cc^\times \in C$ and solve the linear system $L_{\cc^\times}:
  A \xx = \qq$, where the columns of $A$ are
  the points in $C \setminus \{\cc^\times\}$. Again by
  Lemma~\ref{lem:minembr}~\ref{lem:minembr:sublinindep}, the columns 
  of $A$ are
  linearly independent and hence the solution $\xx_{\cc^\times}$ to 
  $L_{\cc^\times}$ is unique, if it exists. If 
  $\xx_{\cc^\times} \geq \0$, we have found the desired coefficients.
  By Theorem~\ref{thm:caratheodory}, there exists a proper subset 
  $C'$ of $C$ that
  ray-embraces $\qq$ and thus there exists a point 
  $\cc^\star \in C$ for which
  $\xx_{\cc^\star} \geq \0$. Solving the linear system 
  $L_{\cc^\times}$ takes $\Oh{d^3}$
  time for each point $\cc^\times \in C$ with Gaussian elimination, 
  and hence we need
  $\Oh{d^4}$ time in total until finding the $\qq$-embracing subset $C
  \setminus \{\cc^\star\}$
  together with the coefficients of the positive combination.
\end{prf}

We can now combine the previous results to show that given a 
$\0$-embracing set, we can find a minimally $\0$-embracing subset in 
polynomial time together with the coefficients of the convex combination 
of the origin.

\begin{lemma}\label{lem:findminembr}
  Let $C \subset \R^d$ be a $\0$-embracing set of size $n$. Given the
  coefficients of the convex combination of $\0$ with the points in $C$,
  we can find a minimally
  $\0$-embracing subset $C' \subseteq C$ and the coefficients of the convex
  combination of $\0$ with the points in $C'$
  in $\Oh{n^2 + n d^3 + d^4}$ time.
\end{lemma}
\begin{prf}
  First, we apply Lemma~\ref{lem:constr_cara} to obtain an affinely 
  independent subset $C'$ of $C$ that embraces the origin. Then, we 
  iteratively test for each point $\cc \in C'$ whether the set 
  $C' \setminus \{\cc\} $ is linearly dependent. If so, we remove 
  $\cc$ from $C'$. After iterating through all
  points, the resulting set still embraces the origin by
  Lemma~\ref{lem:lindepembracing} and moreover, since no proper subset is
  linearly dependent, it is minimally $\0$-embracing.

  The initial application of Lemma~\ref{lem:constr_cara} needs 
  $\Oh{n^2 + n d^3}$ time. Then, checking for one point $\cc \in C'$ 
  whether $C' \setminus \{\cc\}$ is linearly dependent takes 
  $\Oh{d^3}$ time with Gaussian elimination. Because
  $C'$ is affinely independent, we have $|C'|\leq d+1$ and thus the claimed
  running time follows.
\end{prf}

Let now $Q \subset \R^d$ be a set and let $C \subset \R^d$ be a
$\0$-embracing set, as before.  We say a subset $C'$ of $C$ is 
\emph{$\0$-embracing equivalent} to $Q$ with respect to
$C$ if $(C \setminus C') \cup Q$ embraces $\0$.
In the following, we show that if $Q$ embraces the origin when orthogonally
projected onto $\lsp(C)^\perp$, there is always at least one point in 
$C$ that is $\0$-embracing equivalent to $Q$. See 
Figure~\ref{fig:dimreduct:minembreq}.

\begin{figure}[htbp]
  \begin{center}
    \includegraphics{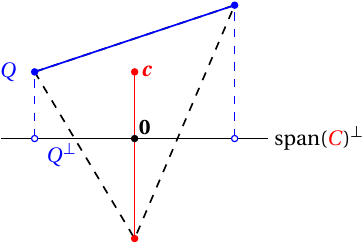}
  \end{center}
  \caption{An example of Lemma~\ref{lem:dimreduct}. The red points 
  constitute the minimal $\0$-embracing set $C$
  and the blue points constitute the set $Q$ that embraces the origin when
  projected onto $\lsp(C)^\perp$. The point $\cc \in C$ is $\0$-embracing
  equivalent to $Q$.}
  \label{fig:dimreduct:minembreq}
\end{figure}

\begin{figure}[htbp]
  \begin{center}
    \includegraphics[page=2]{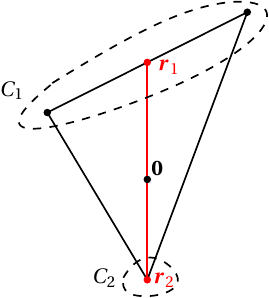}
  \end{center}
  \caption{An example of Lemma~\ref{lem:representatives}. The set 
  $C$ consists of the vertices of the simplex, and the two 
  representative points are with respect to the indicated partition.}
  \label{fig:dimreduct:rep}
\end{figure}

\begin{lemma}\label{lem:dimreduct}
  Let $C \subset \R^d$ be a $\0$-embracing set and let $Q$
  be a set whose orthogonal projection $Q^\perp$ onto $\lsp(C)^\perp$ 
  embraces $\0$. Then, there exists a point $\cc \in C$ that is 
  $\0$-embracing equivalent to $Q$ with respect to $C$. Furthermore, 
  if both $C$ and $Q^\perp$ are minimally $\0$-embracing, we can 
  compute $\cc$ together with the coefficients
  of the convex combination of $\0$ with the points in 
  $(C \setminus \{\cc\}) \cup Q$ in $\Oh{d^4}$ time.
\end{lemma}
\begin{prf}
  We first prove that there is always a point in $C$ that is $\0$-embracing
  equivalent to $Q$. After that, we show how to find this point efficiently.
  We can assume without loss of generality that $C$ is minimally 
  $\0$-embracing, since otherwise the statement holds trivially. 
  Let now $\qq_1, \dots, \qq_m
  \in \R^d$ denote the points in $Q$ and write each
  $\qq_i$, $i \in [m]$, as the sum of a vector $\pp_i \in \lsp(C)$ and a 
  vector $\pp_i^\perp \in \lsp(C)^\perp$. Because $Q$ projected onto 
  $\lsp(C)^\perp$ is $\0$-embracing, there are coefficients 
  $\lambda_1,\dots,\lambda_m \in \R_+$ that sum to $1$ such that 
  $\0 = \sum_{i=1}^{m} \lambda_i\pp_i^\perp$.
  Consider the convex combination 
  $\qq = \sum_{i=1}^{m} \lambda_i \qq_i$ of the
  points in $Q$ with the same coefficients. Since
  \[
  \qq= \sum_{i=1}^{m} \lambda_i \lt(\pp_i + \pp^\perp_i\rt) =
  \lt(\sum_{i=1}^{m} \lambda_i \pp_i \rt) +\lt(\sum_{i=1}^{m} \lambda_i
  \pp^\perp_i\rt) =
  \sum_{i=1}^{m} \lambda_i \pp_i,
  \]
  the point $\qq$ is contained in $\lsp(C)$. By Lemma~\ref{lem:minembr},
  we have $\pos(C) = \lsp(C)$ and hence $-\qq$ is
  ray-embraced by $C$. Now, the cone version of 
  Theorem~\ref{thm:caratheodory} states
  that there is a linearly independent subset $C'$ of $C$ that 
  ray-embraces $-\qq$.
  Because $\dim C = |C|-1$ by Lemma~\ref{lem:minembr},
  the set $C'$ must be a proper subset. Then, $Q$ is $\0$-embracing
  equivalent to all points in $C \setminus C' \neq \emptyset$.

  It remains to show how to find a point in $C \setminus C'$. Recall that
  we assume that both $C$ and $Q^\perp$ are minimally $\0$-embracing, 
  where $Q^\perp$ is the orthogonal projection of $Q$ onto 
  $\lsp(C)^\perp$. Using the algorithm from Lemma~\ref{lem:minembcoeff}, 
  we compute the coefficients of the convex combination
  of the origin with the points in $Q^\perp$ and hence the point $-\qq$ in
  $\Oh{d^4}$ time. Applying Lemma~\ref{lem:minembcoeff} again,
  we can determine the coefficients of the positive
  combination of $-\qq$ with the points in $C$ in $\Oh{d^4}$ time. 
  Similar to the algorithm from Lemma~\ref{lem:findminembr}, we  
  try all $(|C|-1)$-subsets of $C$ until
  we find the linearly independent subset of $C$ that 
  ray-embraces $-\qq$.
  Since the linear combination of $-\qq$ is unique, we thus obtain the
  minimally $(-\qq)$-ray-embracing
  subset $C'$ of $C$ in $\Oh{d^4}$ time. 
  Then, we can choose any point in $C \setminus C'$ as $\cc$. 
  Finally, since we know the coefficients of the convex
  combination of $\qq$ with the points in $Q$ and since we can apply
  Lemma~\ref{lem:minembcoeff} to compute the coefficients of the 
  positive combination
  of $-\qq$ with the points in $C'$, we can compute the coefficients of the
  convex combination of the origin with the points in $C' \cup Q$ by 
  rescaling appropriately. The algorithm takes in total $\Oh{d^4}$ time, 
  as claimed.
\end{prf}

Lemma~\ref{lem:dimreduct} by itself does not yet yield a 
nontrivial approximation algorithm. This is due to the weak guarantee 
that only a single point in $C$ is
$\0$-embracing equivalent to $Q$. To amplify the number of points that can 
be replaced, we conclude this section by showing how to compute a
set of \emph{representative points} $R$ for $C$. Each representative
point stands for a specific subset of $C$ such that if a point in $R$ is
$\0$-embracing equivalent to a set $Q$ with respect to $R$, then the
corresponding subset of $C$ is $\0$-embracing equivalent to $Q$ with 
respect to $C$. See Figure~\ref{fig:dimreduct:rep}.

\begin{lemma}\label{lem:representatives}
  Let $C \subset \R^d$ be a minimally $\0$-embracing set and let
  $C_1,\dots,C_m$ be a partition of $C$ into $m \geq 2$ sets with 
  $|C_i|\geq 1$, for all $i \in [m]$.
  Then, we can compute in $\Oh{d^4}$ time
  a set of points $R= \{\rr_1,\dots,\rr_m\}\subset\R^d$ with the following
  properties:
  \begin{enumerate}[label=(\roman{enumi})]
    \item $R$ is minimally $\0$-embracing.
    \item Let $Q \subset \R^d$ be a set that is $\0$-embracing equivalent to
    some point $\rr_j \in R$ with respect to $R$. Then, $Q$ is 
    $\0$-embracing equivalent to $C_j$ with respect to $C$.
  \end{enumerate}
We call the points in $R$ \emph{representative points} for $C$ with respect
to the partition $C_1,\dots,C_m$.
\end{lemma}

\begin{prf}
  Since $C$ is minimally $\0$-embracing, we can
  write $\0$ as a convex combination $\sum_{\cc \in C} \lambda_\cc \cc$ such
  that all $\lambda_\cc$ are strictly greater than $0$ and sum to $1$. 
  With the algorithm from Lemma~\ref{lem:minembcoeff}, we can compute 
  these coefficients in $\Oh{d^4}$ time. For $i \in [m]$,
  set $\rr_i$ to $\sum_{\cc \in C_i} \lambda_\cc \cc$. Clearly, $R$ is
  $\0$-embracing. Moreover, for all $j \in [m]$, the set 
  $\lt\{\rr_i \midd i \in [m],\, i\neq j\rt\}$ is not $\0$-embracing 
  since otherwise the set $\bigcup_{i=1,\, i\neq j}^m C_i$, a strict 
  subset of $C$, is $\0$-embracing, a contradiction to
  $C$ being minimally $\0$-embracing. Let now $Q$ be a set that
  is $\0$-embracing equivalent to some point
  $\rr_j \in R$ with respect to $R$. That is, the set 
  $Q \cup \lt(R \setminus
  \{\rr_j\}\rt)$ embraces the origin. 
  Because $\rr_i \in \poss{C_i}$, for $i \in
  [m]$, then the set $Q \cup \left(\bigcup_{i=1,\, i\neq j}^m
  C_i\right)$ is $\0$-embracing as well, and hence $Q$ is $\0$-embracing
  equivalent to $C_j$ with respect to $C$.
\end{prf}

\section{$k$-Colorful Choices}
\label{sec:kcol}
Lemmas~\ref{lem:dimreduct} and~\ref{lem:representatives} give
rise to a simple approximation algorithm.
Let $C_1,\dots,C_m \subset \R^d$ be $m$ color classes
that each embrace the origin, and set $k=\max\lt(d-m+2,
\lt\lceil\frac{d+1}{2}\rt\rceil\rt)$. Then, the following algorithm 
recursively computes a $\0$-embracing $k$-colorful choice.
First, we prune $C_1$ with Lemma~\ref{lem:findminembr} and 
partition it into two sets
$C',\,C''$ of size at most $\lt\lceil
(d+1) / 2\rt\rceil$. Using Lemma~\ref{lem:representatives}, we compute two
representative points $\rr',\, \rr''$ for this partition of $C_1$.
Then, we project the remaining $m-1$ color classes onto the
$(d-1)$-dimensional space that is orthogonal to 
$\lsp(\rr',\rr'')^\perp$, and we recursively compute a 
$\0$-embracing $k$-colorful choice $Q$ with respect to
the projections of $C_2,\dots,C_m$. By Lemmas~\ref{lem:dimreduct}
and~\ref{lem:representatives}, one of the two sets $C'$, $C''$, 
say $C'$, is $\0$-embracing
equivalent to $Q$ with respect to $C_1$. Since $Q$ is a
$k$-colorful choice that does not contain points from $C_1$ and since
$|C'|,|C''|\leq k$, the set $C'' \cup Q$ is a $\0$-embracing $k$-colorful
choice. The recursion stops once only one color class is left. Then, 
we are in dimension $d-m+1$. Since $d-m+2 \leq k$, pruning the 
single remaining color class with
Lemma~\ref{lem:findminembr} results already in a $\0$-embracing 
$k$-colorful choice.  For details, see Algorithm~\ref{alg:simapx}.

\begin{alg}
  \SetKwFunction{Recurse}{recurse}
  \KwIn{$m$ sets $C_1,\ldots,C_{m} \subset \R^{d}$ that each embrace 
  the origin, and for each $C_i$, $i \in [m]$, the coefficients 
  of the convex combination of $\0$ with the points in $C_i$}
  \KwOut{minimally $\0$-embracing 
  $\max\lt(d-m+2, \lt\lceil\frac{d+1}{2}\rt\rceil\rt)$-colorful choice}
  $C \gets $ prune $C_1$ with Lemma~\ref{lem:findminembr}\;
  \lIf{$m = 1$}{\Return{$C$}}
  $C',\, C'' \gets $ partition of $C$ into two sets, each of size at most
  $\lt\lceil\frac{d+1}{2}\rt\rceil$\;
  Compute representative points $\rr',\, \rr''$ for $C',\, C''$\;
  $\DOWN{C}_2, \dots, \DOWN{C}_m \gets$ orthogonal projection of
  $C_2,\dots,C_m$ onto $\lsp(\rr',\rr'')^\perp$\;
  $\DOWN{Q} \gets $ \Recurse{$\DOWN{C}_2, \dots, \DOWN{C}_m$}\;
  $Q \gets $ replace projected points in $\DOWN{Q}$ by original points from
  $\bigcup_{i=2}^m C_i$\;
  Determine which point $\rr^{\times} \in \{\rr',\rr''\}$ is $\0$-embracing
  equivalent to $Q$ with Lemma~\ref{lem:dimreduct} 
  and let $C^\times$ be the corresponding subset of $C$\;
  \Return{$\lt(C \setminus C^{\times}\rt) \cup Q$ \textup{pruned with
      Lemma~\ref{lem:findminembr}}}\;
  \caption[Computing $\0$-embracing $(\lceil d/2\rceil +1)$-colorful
  choices.]{Simple Approximation}
  \label{alg:simapx}
\end{alg}

\begin{theorem}\label{thm:simpledimreduct}
  Let $C_1, \ldots, C_m\subset \R^d$ be
  $m\leq d$ color classes such that $C_i$ is a $\0$-embracing set
  of size $\Oh{d}$, for $i \in [m]$. On input $C_1,\dots,C_m$ and given the
  coefficients of the convex combination of the origin for each set $C_i$,
  Algorithm~\ref{alg:simapx} computes a $\0$-embracing $\max\lt(d-m+2,
  \lt\lceil\frac{d+1}{2}\rt\rceil\rt)$-colorful choice in $\Oh{d^5}$ time.
  In particular, for $m=\lt\lfloor d / 2 \rt\rfloor + 1$, the algorithm 
  computes a $(\lceil d/2\rceil +1)$-colorful choice.
\end{theorem}
\begin{prf}
  The correctness of Algorithm~\ref{alg:simapx} is a direct consequence of
  Lemmas~\ref{lem:dimreduct} and~\ref{lem:representatives}. It remains to
  analyze the running time. In each step of the recursion except for the 
  last one, we prune two times a set of size $\Oh{d}$ with
  Lemma~\ref{lem:findminembr}. This needs $\Oh{d^4}$ time. Furthermore, by
  Lemma~\ref{lem:representatives}, computing two representative points 
  also takes $\Oh{d^4}$ time. Finally, given the set $Q$, determining which
  representative point is $\0$-embracing equivalent to $Q$ takes also 
  $\Oh{d^4}$ by Lemma~\ref{lem:dimreduct} and using the fact that
  the recursively computed solution is minimally embracing. Thus, we 
  need $\Oh{d^4}$ time per step of the recursion and there are 
  $\Oh{d}$ recursion steps in total. The total running time is 
  $\Oh{d^5}$.
\end{prf}

Although nontrivial, the fact that we can take in polynomial time 
half of the points from each color class to construct a $\0$-embracing
$\lt(\lt\lceil d/2\rt\rceil+1\rt)$-colorful choice may not be too 
surprising.  In the remainder of this section, we present a 
generalization of Algorithm~\ref{alg:simapx} that computes 
$\0$-embracing $\lceil \eps d\rceil$-colorful choices in 
polynomial time for any fixed $\eps > 0$. The improved approximation
guarantee is achieved by repeatedly replacing subsets of $C$ with
Lemmas~\ref{lem:dimreduct} and~\ref{lem:representatives} in each 
step of the recursion. To still ensure polynomial running time, we 
reduce the dimensionality by a constant fraction in each step of 
the recursion. Additionally, we slightly worsen the desired 
approximation guarantee in each level of the recursion, i.e., if the
current recursion level is $j$ and the dimensionality is $d'$, then 
we do not compute an $\lt\lceil \eps d'\rt\rceil$-colorful
choice, but a $\lt\lceil (1-\eps/2)^{-j/2} \eps d'\rt\rceil$-colorful
choice. As we will see, this additional ``slack'' in the approximation 
guarantee limits the recursion depth to a constant depending only on 
$\eps$. 

In more detail, let $C_1,\dots,C_{d+1}\subset\R^d$ be $d+1$ sets that each
embrace the origin, and let $\eps > 0$ be a parameter. We want to compute 
an $\lt\lceil \eps d\rt\rceil$-colorful choice that embraces the origin. 
Set
\[
d_j = \lt\lceil \lt(1-\frac{\eps}{2}\rt)^j d \rt\rceil \text{~and }
k_j = \lt\lceil \eps \lt(1-\frac{\eps}{2}\rt)^{j/2}
d\rt\rceil,
\]
for $j \in \N$. The sequence $d_j$ controls the dimension reduction
argument with Lemmas~\ref{lem:dimreduct} and~\ref{lem:representatives}, 
i.e., in the $j$th recursion level, the dimensionality of the input 
will be $d_j$. The sequence $k_j$ defines the approximation guarantee 
in the $j$th recursion level.  Note that $d_0 = d$ and 
$k_0 = \lt\lceil \eps d \rt\rceil$. Assume now we are in recursion 
level $j$. That is, the input consists of $d_j+1$
color classes $C_1,\dots,C_{d_j+1} \subset \R^{d_j}$ that each embrace the
origin together with the coefficients of their convex combinations 
of the origin.
We want to compute a $\0$-embracing $k_j$-colorful choice.
As in the previous algorithm, we begin by computing a minimal 
$\0$-embracing subset $C$ of $C_1$ with Lemma~\ref{lem:findminembr}. 
If $k_j \geq d_j + 1$, then $C$ is already a valid approximation. 
Otherwise, we iteratively transform $C$ into a $k_j$-colorful choice. 
For this, we repeatedly replace subsets
of $C$ with points from $C_2 \cup \dots \cup C_{d_j + 1}$
until it contains at most $k_j$ points from each color. This 
is done as follows. Set $m = d_j - d_{j+1} + 1$. 
In the general situation, $C$ contains points from several
color classes, and we partition $C$ into sets $D_1,\dots,D_m$ by
distributing the points from each color in $C$ equally among 
these $m$ sets. Then,
we compute representative points $\rr_1,\dots,\rr_m$ for this partition. 
Let
$C^\star_1,\dots,C^\star_{d_{j+1}+1} \in \lt\{C_2,\dots,C_{d_j+1}\rt\}$ be
$d_{j+1}+1$ color classes, where we discuss shortly how they are chosen.
We recursively compute a $k_{j+1}$-colorful choice $Q$ for
$C^\star_1,\dots,C^\star_{d_{j+1}+1}$ that embraces the
origin when projected on $U = \lsp(\rr_1,\dots,\rr_m)^\perp$. 
Note that $\dim U = d_j-(m-1) = d_{j+1}$
and hence the dimensionality of the input in recursion level 
$j+1$ is $d_{j+1}$, as desired. Then, by Lemmas~\ref{lem:dimreduct}
and~\ref{lem:representatives}, at least one representative 
point $\rr_{i^\times}$ and
hence at least one of the sets $D_{i^\times}$ is
$\0$-embracing equivalent to $Q$. We set $C$ to $\lt(C \setminus
D_{i^\times}\rt) \cup Q$ and prune it with Lemma~\ref{lem:findminembr}. 
We repeat these steps until $C$ is a $k_j$-colorful choice. 

To ensure progress, $m$ should be smaller than $k_j$ so that 
$D_{i^\times}$ is guaranteed to contain a point from each 
color that appears more than $k_j$ times
in $C$. Furthermore, $Q$ should not contain points with colors that appear
``often'' in $C$. We call a color class $C_i$ \emph{light} with 
respect to $C$ if $|C \cap C_i| \leq k_j - k_{j+1}$, and 
\emph{heavy}, otherwise. For the recursion, we use only light 
color classes. A $k_{j+1}$-colorful choice with light colors can be 
added safely to $C$ without increasing any color over the
threshold $k_j$. In particular, since we start with $C=C_1$ and 
use only light color classes, no other color class can ever occur 
more than $k_j$ times in $C$ and hence we are finished once the 
number of points from $C_1$ is at most $k_j$.
Please refer to Algorithm~\ref{alg:mainapx} for details.

\begin{alg}
  \SetKwFunction{Recurse}{recurse}
  \KwIn{recursion depth $j\in\N_0$ (initially $0$), original 
  dimension $d \in \N$, approximation parameter $\eps > 0$, $d_j+1$ sets
  $C_1,\ldots,C_{d_j+1} \subset \R^{d_j}$ that each embrace the origin, 
  and for each $C_i$ the coefficients of the convex
  combination of $\0$ with the points in $C_i$
  }
  \KwOut{minimally $\0$-embracing $k_j$-colorful choice}
  $k_j \gets \lt\lceil \eps \lt(1-\frac{\eps}{2}\rt)^{j/2} d \rt\rceil$\;
  $d_{j+1} \gets \lt\lceil \lt(1-\frac{\eps}{2}\rt)^{j+1} d \rt\rceil$\;
  $m \gets d_j - d_{j+1} + 1$\;
  $C \gets $ prune $C_1$ with
  Lemma~\ref{lem:findminembr}\;\label{alg:mainapx:initprune}
  \While{$|C \cap C_1| > k_j$}{\label{alg:mainapx:while}
    $D_1,\dots, D_m \gets $ partition of $C$ s.t.\ the
    points from each color class are evenly 
    distributed\;\label{alg:mainapx:part}
    Compute representative points $\rr_1,\dots,\rr_m$ for
        $D_1,\dots,D_m$ with
        Lemma~\ref{lem:representatives}\;\label{alg:mainapx:representatives}
    Find $d_{j+1} + 1$ light color classes 
    $C^\star_1,\dots,C^\star_{d_{j+1} + 1} \in 
    \lt\{C_2,\dots,C_{d_j+1}\rt\}$\;
    $\DOWN{C}_1, \dots, \DOWN{C}_{d_{j+1}+1} \gets$ orthogonal projection of
    $C^\star_1,\dots,C^\star_{d_{j+1}+1}$ onto
        $\lsp(\rr_1,\dots,\rr_m)^\perp$\;\label{alg:mainapx:proj}
    $\DOWN{Q} \gets $\Recurse{$j+1$, $d$, $\eps$, $\DOWN{C}_1, \dots,
        \DOWN{C}_{d_{j+1}+1}$}\;\label{alg:mainapx:rec}
    $Q \gets $ replace projected points in $\DOWN{Q}$ by original 
    points from
    $\bigcup_{i=1}^{d_{j+1}+1} C^\star_i$\;
    Determine which point $\rr_{i^\times} \in \{\rr_1,\dots,\rr_m\}$ 
    is $\0$-embracing equivalent to $Q$ with Lemma~\ref{lem:dimreduct}\;
    $C \gets \lt(C \setminus D_{i^\times}\rt) \cup Q $ pruned with
        Lemma~\ref{lem:findminembr}\;
  }
  \Return{$C$}\;\label{alg:mainapx:return}
  \caption[Computing $\0$-embracing $\lt\lceil \eps d\rt\rceil$-colorful
  choices.]{$\lt\lceil \eps d\rt\rceil$-Approximation}
  \label{alg:mainapx}
\end{alg}

The next lemma states that for $\eps$ fixed, the number of necessary 
recursions before a trivial approximation with 
Lemma~\ref{lem:findminembr} suffices is constant.

\begin{lemma}\label{lem:recdepth}
  For any $\eps = \Om{d^{-1/4}}$ there exists a 
  $j = \Th{ \eps^{-1} \ln \eps^{-1}}$ such that $k_j \geq d_j+1$.
\end{lemma}
\begin{prf}
Replacing $d_j$ with its definition, we obtain
\begin{equation}\label{lem:recdepth:dj}
  d_j + 1 = \left\lceil \lt(1-\frac{\eps}{2}\rt)^j d \right\rceil + 1
  \leq \lt(1-\frac{\eps}{2}\rt)^j d + 2.
\end{equation}
Using $\ln\lt(1- \frac{\eps}{2}\rt) \geq  -\eps$ if $\eps\leq 1$, we 
have for
$j \leq \frac{1}{\eps} \ln d$,
\begin{equation}\label{lem:recdepth:geq1}
\lt(1-\frac{\eps}{2}\rt)^j d \geq e^{-\eps j}d \geq 1.
\end{equation}
Furthermore, using that $\ln\lt(1-\frac{\eps}{2}\rt) \leq
-\frac{\eps}{2}$, we have for $j \geq \frac{4}{\eps} \ln \frac{3}{\eps}$
\begin{equation}\label{lem:recdepth:eps}
  3 \lt(1 - \frac{\eps}{2}\rt)^{j/2} \leq 
  3 e^{-\eps j/4} \leq \eps.
\end{equation}
Combining~\eq{lem:recdepth:geq1} and~\eq{lem:recdepth:eps}
with~\eq{lem:recdepth:dj}, we get
\[
  d_j + 1
  \leq 3 \lt(1 - \frac{\eps}{2}\rt)^j d
  \leq \eps \lt(1 - \frac{\eps}{2}\rt)^{j/2} d
  \leq \lt\lceil\eps \lt(1 - \frac{\eps}{2}\rt)^{j/2} d\rt\rceil
  = k_j.
\]
For $d = \Om{\eps^{-1/4}}$, there is a $j$ with
$\frac{4}{\eps}\ln \frac{3}{\eps} \leq j \leq \frac{1}{\eps} \ln d$. 
The claim follows.
\end{prf}

Next, we show that if the recursion depth is not too large, then we can
always find enough light color classes.

\begin{lemma}\label{lem:light}
  Let $j \in \N$ and let $C_1,\dots,C_{d_j+1}\subset \R^{d_j}$ be 
  $d_j+1$ color classes.
  Furthermore, let $C \subseteq \bigcup_{i=1}^{d_j+1} C_i$ be a set 
  of size at most $d_j+1$. For all $j=\Oh{\eps^{-1} \ln (\eps^3 d)}$, 
  there exist $d_{j+1}+1$ light color classes with respect to $C$.
\end{lemma}
\begin{prf}
We recall that a color class $C_i$, $i \in [d_j+1]$, is light with 
respect to $C$ if $|C \cap C_i| \leq k_j - k_{j+1}$. Then, the 
number of heavy color classes $h$ is bounded by
\begin{equation}\label{lem:light:heavy}
  h \leq \lt\lceil\frac{d_j + 1}{k_j - k_{j+1}}\rt\rceil \leq
  \frac{2 d_j}{k_j - k_{j+1}} + 1,
\end{equation}
since $d_j \geq 1$ for all $j \in \N$. We can bound the denominator 
as follows
\begin{multline}\label{lem:light:denom}
    k_j - k_{j+1}
    = \lt\lceil \eps \lt(1-\frac{\eps}{2}\rt)^{j/2} d\rt\rceil
        - \lt\lceil \eps \lt(1-\frac{\eps}{2}\rt)^{(j+1)/2} d\rt\rceil
    \geq \eps \lt(1-\frac{\eps}{2}\rt)^{j/2} d
        - \eps \lt(1-\frac{\eps}{2}\rt)^{(j+1)/2} d - 1
    \\
    = \eps \lt(1-\frac{\eps}{2}\rt)^{j/2} d 
    \lt( 1 - \sqrt{1-\frac{\eps}{2}}\rt) - 1
    \geq \frac{\eps^2}{4} \lt(1-\frac{\eps}{2}\rt)^{j/2} d - 1,
\end{multline}
where we apply $1 - \sqrt{1 - \frac{\eps}{2}} \geq \frac{\eps}{4}$ in 
the last
inequality.  Using that $\ln\lt(1-\frac{\eps}{2}\rt) \geq -\eps$ if 
$\eps \leq 1$, we have for $j \leq \frac{2}{\eps} \ln \frac{\eps^2 d}{8}$
\begin{equation}\label{lem:light:denomsim}
1 \leq 
\frac{\eps^{2}}{8} e^{-\eps j/2} d \leq
\frac{\eps^2}{8} \lt(1-\frac{\eps}{2}\rt)^{j/2} d
\end{equation}
and hence~\eq{lem:light:denom} can be simplified to
\begin{equation}\label{lem:light:denomfinal}
    k_j - k_{j+1} \geq \frac{\eps^2}{8} \lt(1-\frac{\eps}{2}\rt)^{j/2} d.
\end{equation}
Plugging (\ref{lem:light:denomfinal}) into (\ref{lem:light:heavy}) and
using~\eq{lem:light:denomsim},
we obtain
\begin{equation*}\label{lem:light:heavyfinal}
  h \leq \frac{2 \lt\lceil \lt(1-\frac{\eps}{2}\rt)^j d\rt\rceil}
  {\frac{\eps^2}{8}
      \lt(1-\frac{\eps}{2}\rt)^{j/2} d} + 1
  \leq \frac{2 \lt(1-\frac{\eps}{2}\rt)^j d}{\frac{\eps^2}{8}
      \lt(1-\frac{\eps}{2}\rt)^{j/2} d} + 3
  = \frac{16}{\eps^2} \lt(1-\frac{\eps}{2}\rt)^{j/2} + 3.
\end{equation*}
Then, the number $\ell$ of light color classes is at least
\begin{multline}\label{lem:light:light1}
    \ell = d_j + 1 - h
    \geq \lt\lceil \lt(1-\frac{\eps}{2}\rt)^j d \rt\rceil - 
    \frac{16}{\eps^2}
        \lt(1-\frac{\eps}{2}\rt)^{j/2} - 2
    \\
    \geq \lt(1-\frac{\eps}{2}\rt)^j d \lt(1 - \frac{16}{\eps^2
          \lt(1-\frac{\eps}{2}\rt)^{j/2} d} -
    \frac{2}{\lt(1-\frac{\eps}{2}\rt)^j d} \rt).
\end{multline}
For $j \leq \frac{2}{\eps}\ln \frac{\eps^3 d}{128}$, using
$\ln \lt(1-\frac{\eps}{2}\rt) \geq  -\eps$ if $\eps \leq 1$, we have
\begin{equation*}
    \frac{16}{\eps^2 \lt(1-\frac{\eps}{2}\rt)^{j/2} d} +
    \frac{2}{\lt(1-\frac{\eps}{2}\rt)^j d} 
 \leq    \frac{16}{\eps^2 e^{-\eps j/2} d} +
    \frac{2}{e^{-\eps j/2} d}
    \leq \frac{\eps}{8} + \frac{\eps}{8}  
    \leq \frac{\eps}{4}
\end{equation*}
and thus~\eq{lem:light:light1} implies
\begin{equation}\label{lem:light:lightinterm}
  \ell \geq \lt(1-\frac{\eps}{4}\rt)\lt(1-\frac{\eps}{2}\rt)^j d.
\end{equation}
For $j \leq \frac{2}{\eps}\ln \frac{\eps d}{2}$, using
$\ln \lt(1-\frac{\eps}{2}\rt) \geq  -\eps$ if $\eps \leq 1$, we can bound
\begin{equation}\label{lem:light:epshalf}
  \frac{\eps}{4} \lt(1-\frac{\eps}{2}\rt)^j d 
  \geq \frac{\eps}{4} e^{-\eps j/2} d \geq 2.
\end{equation}
Combining~\eq{lem:light:epshalf} with (\ref{lem:light:lightinterm}), we get
\[
  \ell
  \geq \lt(1-\frac{\eps}{2}\rt)^{j+1} d
      + \frac{\eps}{4} \lt(1 - \frac{\eps}{2}\rt)^j d
  \geq \lt(1-\frac{\eps}{2}\rt)^{j+1} d + 2
  \geq \lt\lceil \lt(1-\frac{\eps}{2}\rt)^{j+1} d \rt\rceil + 1
  = d_{j+1} + 1.
\]
Thus, for $j = \Oh{\eps^{-1} \ln (\eps^3 d)}$, there are at least $d_{j+1}
+1$ light color classes with respect to $C$.
\end{prf}

Before we finally prove correctness, we show if the recursion depth 
$j$ is not too large, then each set of the partition of $C$ contains 
at least one point from $C_1$ until $C$ is a $k_j$-colorful choice. 
This implies that each iteration of
the while-loop decreases the amount of points from $C_1$ in $C$.

\begin{lemma}\label{lem:while}
  For all $j = \Oh{\eps^{-1} \ln (\eps d)}$, we have 
  $m = d_j - d_{j+1} + 1 \leq k_j + 1$.
\end{lemma}
\begin{prf}
  First, we upper bound $m$ as follows:
\begin{equation}\label{lem:while:k}
  \begin{split}
    m = d_j - d_{j+1} + 1
    & = \lt\lceil \lt(1-\frac{\eps}{2}\rt)^j d \rt\rceil
        - \lt\lceil \lt(1-\frac{\eps}{2}\rt)^{j+1} d \rt\rceil + 1
    \\
    & \leq \lt(1-\frac{\eps}{2}\rt)^j d - 
    \lt(1-\frac{\eps}{2}\rt)^{j+1} d + 2
    = \frac{\eps}{2} \lt(1-\frac{\eps}{2}\rt)^j d + 2.
  \end{split}
\end{equation}
For $j \leq  \frac{2}{\eps} \ln \frac{\eps d}{2}$, 
with $\ln\lt(1-\frac{\eps}{2}\rt) \geq -\eps$ if 
$\eps \leq 1$,
we obtain $\frac{\eps}{2}
\lt(1-\frac{\eps}{2}\rt)^j d \geq \frac{\eps}{2} e^{-\eps j/2}d \geq 1$. 
Using this in (\ref{lem:while:k}), we get
\begin{equation*}
    m
  \leq \eps \lt(1-\frac{\eps}{2}\rt)^j d + 1
  \leq \lt\lceil \eps \lt(1-\frac{\eps}{2}\rt)^j d \rt\rceil + 1 = k_j + 1,
\end{equation*}
as desired.
\end{prf}

\begin{theorem}
\label{thm:bapx}
  Let $C_1, \ldots, C_{d+1}\subset \R^d$ be
  $d+1$ sets such that $C_i$ is a $\0$-embracing set
  of size $\Oh{d}$, for $i \in [d+1]$, and let $\eps = \Om{d^{-1/4}}$ be a
  parameter. On input $0$, $d$, $\eps$, $C_1,\dots,C_{d+1}$, and given the
  coefficients of the convex combination of the origin with the points 
  in $C_i$, for $i \in [d+1]$,
  Algorithm~\ref{alg:mainapx} computes a $\0$-embracing 
  $\lceil \eps d\rceil$-colorful
  choice in $d^{\Oh{\eps^{-1} \ln \eps^{-1}}}$ time.
\end{theorem}

\begin{prf}
We begin by showing that if the algorithm enters the while 
loop in recursion level $j$, it is always possible to find $d_{j+1}
+ 1$ light color classes and that the projections
$\DOWN{C}_1,\dots,\DOWN{C}_{d_{j+1}+1}$ of these color
classes are $\0$-embracing subsets of 
$\R^{d_{j+1}}$~(Line~\ref{alg:mainapx:proj}). In other words, we show
that recursion is possible if $C$ is not a $k_j$-colorful choice.
Assume now the algorithm enters the while loop in recursion level $j$.
Then, $C$ is a minimally $\0$-embracing subset of 
$C_1 \subset \R^{d_j}$ and has size at least $k_j+1$. In
Line~\ref{alg:mainapx:part}, we partition $C$ into $m$ sets 
$D_1,\dots,D_m$ by distributing the points from each color class equally. 
By Lemma~\ref{lem:while}, we have $m\leq k_j + 1$, 
for $j = \Oh{\eps^{-1} \ln (\eps d)}$,
and hence each set $D_i$ is nonempty. Thus,
the algorithm from Lemma~\ref{lem:representatives} can be applied in
Line~\ref{alg:mainapx:representatives}
to compute the representative points $\rr_1,\dots,\rr_m$. Moreover $\dim
\lsp\lt(\rr_1,\dots,\rr_m\rt) = m-1$ by Lemma~\ref{lem:representatives} and
Lemma~\ref{lem:minembr}.
Thus, $\dim \lsp\lt(\rr_1,\dots,\rr_m\rt)^\perp = d - m + 1 = d_{j+1}$.
Now, Lemma~\ref{lem:light} guarantees that we can always find $d_{j+1}+1$ 
light color classes $C^\star_1,\dots,C^\star_{d_{j+1}+1}$, 
if $j = \Oh{\eps^{-1} \ln \eps^3 d}$. Because each color 
class $C^\star_i$, $i \in [d_{j+1}+1]$, is
$\0$-embracing, so are their orthogonal projections onto
$\lsp(\rr_1,\dots,\rr_k)^T$. Thus, recursion is possible 
if $j = \Oh{\eps^{-1} \ln \eps^3 d}$. By Lemma~\ref{lem:recdepth}, 
the recursion depth is limited to
$\Th{\eps^{-1} \ln \eps^{-1}}$, since then pruning $C_1$ with
Lemma~\ref{lem:findminembr} in Line~\ref{alg:mainapx:initprune} is already a
$\0$-embracing $k_j$-colorful choice. In this case, the while loop is never
executed. We conclude that for $\eps = \Om{d^{-1/4}}$, recursion is always
possible as long as $C$ is not a $k_j$-colorful choice.

Next, we prove that the algorithm computes in recursion level $j$ a
$\0$-embracing $k_j$-colorful choice. As discussed above, the
recursion terminates after $\Oh{\eps^{-1} \ln \eps^{-1}}$ steps when
the set $C$ from Line~\ref{alg:mainapx:initprune} is already a 
$\0$-embracing $k_j$-colorful choice. If $C$ is not already a 
valid approximation, the while loop is executed. In each 
iteration of the while loop, $C$ is partitioned into $m$ sets
$D_1,\dots,D_m$ by distributing the points from each color 
equally among the $D_i$. By Lemma~\ref{lem:while}, 
$m\leq k_j + 1$ for $j = \Oh{\eps^{-1} \ln \eps d}$
and hence each set $D_i$, $i \in [m]$, contains at least one 
point from $C_1$. Applying Lemmas~\ref{lem:dimreduct} 
and~\ref{lem:representatives}, one of these
sets, say $D_{i^\times}$, is replaced in $C$  by a recursively computed
$k_{j+1}$-colorful choice $Q$
that is $\0$-embracing when projected onto $\lsp(\rr_1,\dots,\rr_m)^\perp$.
Since we use in the recursion only light color classes with
respect to $C$, and since $C_1$ is not a light color class, 
each iteration of the while loop strictly decreases the number of 
points from $C_1$ in $C$.  Moreover, because $Q$ contains 
only points from light color classes and since it
is a $k_{j+1}$-colorful choice, $\lt(C \setminus D_{i^\times}\rt) \cup Q$
contains at most $k_j$ points from the color classes $C_2,\dots,C_{d_j+1}$.
Thus, after $\Oh{d}$ iterations, $C$ is a $\0$-embracing 
$k_j$-colorful choice.

It remains to analyze the running time. The initial computation of $C$ in
Line~\ref{alg:mainapx:initprune} and each iteration of the 
while loop except for the recursive call takes $\Oh{d^4}$ time. 
Since the while loop is executed $\Oh{d}$
times and since the recursion depth is bounded by $\Oh{\eps^{-1} \ln
\eps^{-1}}$, the total running time of Algorithm~\ref{alg:mainapx} is 
$d^{\Oh{\eps^{-1} \ln \eps^{-1}}}$.
\end{prf}

\subsection{Applications}
\label{sec:kcol:applications}

As discussed in the introduction, the main motivation for $k$-colorful
choices is their application in polynomial-time reductions to \CCP. We
begin by presenting the proofs whose interpretation as algorithms 
results in the polynomial reductions. Then, we give precise 
bounds on the quality of the
obtained approximation algorithms for \Centerpoint, \Tverberg, and
\ColKirchberger when having access to an algorithm that on input $d+1$
color classes $C_1,\dots,C_{d+1}$, each $\0$-embracing and of size at 
most $d+1$, computes a $\0$-embracing $k(d)$-colorful choice in time $W(d)$.

\begin{theorem}[{Centerpoint 
theorem~\cite[Theorem~1]{Rado1946}}]\label{thm:centerpoint}
Let $P \subset \R^d$ be a point set. Then, there exists a point $\qq \in
\R^d$ such that for any halfspace $h^-$ with $\qq \in h^-$, we 
have $|P \cap h^-| \geq \left\lceil\frac{|P|}{d+1}\right\rceil$. \qed
\end{theorem}

Teng~\cite[Theorem~8.4]{Teng1991} showed that given a point set 
$P \in \R^d$ and a candidate centerpoint $\qq \in \R^d$, it is 
\coNP-complete to decide whether $\qq$ is
a centerpoint of $P$, if $d$ is part of the input. 
For $d=1$, a centerpoint is equivalent to a median of a set of 
numbers and hence can be computed in $\Oh{|P|}$ 
time~\cite{BlumFlPrRiTa1973}. Jadhav and Mukhopadhyay~\cite{JadhavMu1994} 
showed that linear time is sufficient even in two dimensions. For 
$d\geq 3$ fixed, the best known algorithm is by Chan~\cite{Chan2004} 
who showed how to compute a point with maximum Tukey depth,
a stronger notion than being a centerpoint, in expected time 
$\Oh{n^{d-1}}$.

Although it is in general \coNP-complete to verify centerpoints, Tverberg
partitions serve as
polynomial-time checkable certificates for a subset of centerpoints.
In recent years, this property has been exploited algorithmically to derive
efficient approximation algorithms for
centerpoints~\cite{MulzerWe2013,MillerSh2010}.
The existence of Tverberg points is guaranteed by Tverberg's
theorem~\cite{Tverberg1966}.

\begin{theorem}[Tverberg's theorem~\cite{Tverberg1966}]\label{thm:tverberg}
  Let $P \subset \R^d$ be a point set of size $n$. Then, there 
  always exists a
  Tverberg $\left\lceil\frac{|P|}{d+1}\right\rceil$-partition for $P$.
  Equivalently, let $P$ be of size $(m-1)(d+1)+1$, with $m \in \N$. 
  Then, there exists a Tverberg $m$-partition for $P$.
\end{theorem}

While Tverberg's first proof is quite involved, several
simplified subsequent 
proofs~\cite{Tverberg1981,TverbergVr1993,Sarkaria1992,Roudneff2001} 
have been published.
Here, we present Sarkaria's proof~\cite{Sarkaria1992} with further
simplifications by \Barany and Onn~\cite{BaranyOn1997} and Arocha
\etal~\cite{ArochaBaBrFaMo2009}.
The main tool is the next lemma that establishes a correspondence 
between the intersection of convex hulls of low-dimensional point sets
and the embrace of the origin of certain high-dimensional
point sets. It was extracted from Sarkaria's proof by Arocha
\etal~\cite{ArochaBaBrFaMo2009}. In the following, we denote with 
$\otimes$ the \emph{tensor product} that maps two points 
$\pp \in \R^d$, $\qq \in \R^m$ to the point
\[
\pp \otimes \qq =
\begin{pmatrix}
  (\qq)_1 \pp \\
  (\qq)_2 \pp \\
  \vdots \\
  (\qq)_m
  \pp
\end{pmatrix} \in \R^{dm},
\]
where $(\qq)_i \pp$ denotes the vector $\pp$ scaled by the $i$th component
of $\qq$, for $i \in [m]$.
Then, $\otimes$ is bilinear, i.e., for all $\pp_1,\pp_2 \in
\R^d$, $\qq \in \R^m$, and $\alpha_1,\alpha_2 \in \R$, we have
\[
\lt(\alpha_1 \pp_1 + \alpha_2 \pp_2\rt) \otimes \qq
= \alpha_1 \lt(\pp_1 \otimes \qq\rt) + \alpha_2 \lt(\pp_2 \otimes \qq\rt)
\]
and similarly, for all $\pp \in \R^d$, $\qq_1,\qq_2 \in \R^m$, and
$\alpha_1,\alpha_2 \in \R$, we have
\[
\pp \otimes \lt(\alpha_1 \qq_1 + \alpha_2 \qq_2\rt)
= \alpha_1 \lt(\pp \otimes \qq_1\rt) + \alpha_2 \lt(\pp \otimes \qq_2\rt).
\]

\begin{lemma}[Sarkaria's 
lemma~\cite{Sarkaria1992},~{\cite[Lemma~2]{ArochaBaBrFaMo2009}}]
\label{lem:sarkaria}
Let $P_1,\dots,P_m \subset \R^d$ be $m$ point sets and 
let $\qq_1,\dots,\qq_{m}
\subset \R^{m-1}$ be $m$ vectors with $\qq_i = \e_i$ for $i \in [m-1]$ and
$\qq_m = -\1$. For $i \in [m]$, we define
\[
  \UP{P}_i = \set{\TwoRowVec{\pp}{1} \otimes \qq_i
  \midd \pp \in P_i}\subset \R^{(d+1) (m-1)}.
\]
Then, the intersection of the convex hulls $\bigcap_{i=1}^m \convv{P_i}$ 
is nonempty if and only if $\;\bigcup_{i=1}^m \UP{P}_i$ embraces the origin.
\end{lemma}
\begin{prf}
  Assume there is a point $\pp^\star \in \bigcap_{i=1}^m \convv{P_i}$.
  There exist coefficients
  $\lambda_{i,\pp} \in \Rp$ that sum to $1$ such that
  $\pp^\star = \sum_{\pp \in P_i} \lambda_{i,\pp}\pp$.
  Consider the points $\up{\pp}_i \in \convv{\UP{P}_i}$, $i \in [m]$, 
  that we obtain by using the same convex coefficients for the points 
  in $\UP{P}_i$, i.e., set
\[
  \up{\pp}_{i}
  = \sum_{\pp \in P_i} \lambda_{i,\pp}
          \lt(\TwoRowVec{\pp}{1} \otimes \qq_i\rt) \in \convv{\UP{P}_i}.
\]
We claim that $\sum_{i=1}^{m} \up{\pp}_i = \0$ and thus
$\0 \in \convv{\bigcup_{i=1}^{m} \UP{P}_i}$. Indeed, we have
\begin{multline*}
  \sum_{i=1}^{m} \up{\pp}_i = 
  \sum_{i=1}^{m} \sum_{\pp \in P_i} \lambda_{i,\pp}
  \left( \TwoRowVec{\pp}{1} \otimes \qq_i \right)
  \\
  = \sum_{i=1}^{m} \left( \sum_{\pp \in P_i} \lambda_{i,\pp}
  \TwoRowVec{\pp}{1} \right) \otimes \qq_i
  = \sum_{i=1}^{m} \TwoRowVec{\pp^\star}{1} \otimes \qq_i
  \\
  = \TwoRowVec{\pp^\star}{1} \otimes \left( \sum_{i=1}^{m} \qq_i \right)
  = \TwoRowVec{\pp^\star}{1} \otimes \0 = \0,
\end{multline*}
using the bilinearity of  $\otimes$.

Assume now that $\bigcup_{i=1}^m \UP{P}_i$ embraces the origin. 
We want to show that $\bigcap_{i=1}^m \convv{P_i}$ is nonempty. 
Then, we can express the origin as a convex combination $\sum_{i=1}^{m}
\sum_{\up{\pp} \in \UP{P}_i} \lambda_{i,\up{\pp}} \up{\pp}$ with
$\lambda_{i,\up{\pp}} \in\Rp$ for $i \in
[m]$ and $\up{\pp} \in \UP{P}_i$, and $\sum_{i=1}^{m} \sum_{\up{\pp} \in
\UP{P}_i} \lambda_{i,\up{\pp}} = 1$. Hence, we have
\[
\0 = \sum_{i=1}^{m} \sum_{\up{\pp} \in \UP{P}_i} \lambda_{i,\up{\pp}}
\left( \TwoRowVec{\pp}{1} \otimes \qq_i \right)
 = \sum_{i=1}^{m} \left( \sum_{\up{\pp} \in \UP{P}_i} \lambda_{i,\up{\pp}}
 \TwoRowVec{\pp}{1} \right) \otimes \qq_i,
\]
again using the bilinearity of $\otimes$.
By the choice of $\qq_1,\dots,\qq_m$, there is (up to multiplication 
with a scalar) exactly one linear dependency: $\0 = \sum_{i=1}^{m} \qq_i$.
Thus,
\[
\sum_{\up{\pp} \in \UP{P}_1} \lambda_{1,\up{\pp}} \TwoRowVec{\pp}{1}
= \dots =
\sum_{\up{\pp} \in \UP{P}_m} \lambda_{m,\up{\pp}} \TwoRowVec{\pp}{1} =
\TwoRowVec{\pp^\star}{c},
\]
where $\pp^\star \in \R^d$ and $c \in \R$. In particular, the last equality
implies that
\[
\sum_{\up{\pp} \in \UP{P}_1} \lambda_{1,\up{\pp}}
= \dots =
\sum_{\up{\pp} \in \UP{P}_m} \lambda_{m,\up{\pp}}
= c.
\]
Now, as for all $i \in [m]$ and $\up{\pp} \in \UP{P}_i$, the coefficient
$\lambda_{i,\up{\pp}}$ is nonnegative and as the sum $\sum_{i \in [m]}
\sum_{\up{\pp} \in \UP{P}_i} \lambda_{i,\up{\pp}}$ is $1$, we must have 
$c = 1/m \in (0,1]$. Hence,
the point $m \pp^\star$ is common to all convex
hulls $\convv{P_1}$, $\ldots$, $\convv{P_m}$.
\end{prf}

Please refer to Figure~\ref{fig:sarkaria} for an example of Sarkaria's 
lifting argument.  Little work is now left to obtain Tverberg's theorem 
from Lemma~\ref{lem:sarkaria} and the colorful \Caratheodory theorem.
\begin{figure}[htbp]
  \begin{center}
    \includegraphics{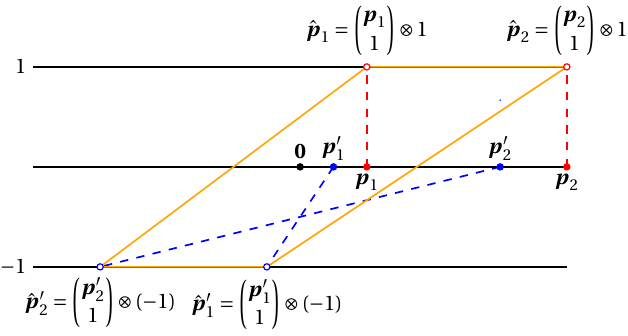}
  \end{center}
  \caption[Example of Sarkaria's Lemma.]{An example of Sarkaria's lemma for
  $d=1$ and $m=2$. The set $P_1$
  consists of the red points and the set $P_2$ consists of the blue points.
  Since the convex hulls of $P_1$ and $P_2$ intersect, the lifted points 
  embrace the origin.}
  \label{fig:sarkaria}
\end{figure}

\begin{prf}[Proof of Theorem~\ref{thm:tverberg}]\label{thm:tverberg:proof}
Let $P = \set{\pp_1,\dots,\pp_n} \subset \R^d$ be a point set of size
$n=(d+1)(m-1)+1$ and let $P_1,\dots,P_m$ denote $m$ copies of $P$.
For each set $P_j \subset \R^d$, $j \in [m]$, we construct a $((d+1)
(m-1))$-dimensional set $\UP{P}_j$ as in Lemma~\ref{lem:sarkaria}, i.e.,
\[
  \UP{P}_j = \set{ \up{\pp}_{i,j} = 
  \TwoRowVec{\pp_i}{1} \otimes \qq_j \midd \pp_i \in P}
  \subset \R^{(d+1) (m-1)} = \R^{n-1}.
\]
For $i \in [n]$, we denote with 
$\UP{C}_i \subseteq \bigcup_{j=1}^m \UP{P}_j$
the set of points $\set{\up{\pp}_{i,j}  \midd j \in [m]}$ that 
correspond to $\pp_i \in P$,
and we color these points with color $i$. For $i \in [n]$, note that
Lemma~\ref{lem:sarkaria} applied to $m$ copies of the singleton set 
$\set{\pp_i}
\subseteq P$ guarantees that the color class $\UP{C}_i \in
\R^{n-1}$ embraces the origin. Hence, we have $n$ color classes
$\UP{C}_1,\dots,\UP{C}_n$ that embrace the origin in $\R^{n-1}$. Now,
by Theorem~\ref{thm:colcara}, there is a
colorful choice $\UP{C} = \set{\up{\cc}_1,\dots,\up{\cc}_n} \subseteq
\bigcup_{i=1}^n \UP{C}_i$ with $\up{\cc}_i \in \UP{C}_i$ that embraces the
origin, too. Because $\UP{C}$ embraces the origin,
Lemma~\ref{lem:sarkaria} guarantees that the convex hulls of the sets 
$T_j = \set{ \pp_i \in P \midd \up{\pp}_{i,j} \in \UP{C}}$, $j \in [m]$, 
have a point in common.
Moreover, since all points in $\bigcup_{j=1}^m \UP{P}_j$
that correspond to the same point in $P$ have the same
color, each point $\pp_i \in P$ appears in exactly one set 
$T_j$, $j \in [m]$.
Thus, $\mc{T} =\set{T_1,\dots, T_m}$ is a Tverberg $m$-partition of $P$.
\end{prf}

Even less effort is required to obtain the colorful Kirchberger theorem 
from Lemma~\ref{lem:sarkaria}. Let $A, B \subset\R^d$ be two point sets.
Kirchberger's theorem~\cite{Kirchberger1903} states that if for all 
subsets $C \subset A \cup B$ of size at most $d+2$,
the sets $\convv{A\cap C}$ and $\convv{B \cap C}$ have an empty 
intersection, then $\convv{A}$ and $\convv{B}$ have an empty intersection.
Arocha~\etal~\cite{ArochaBaBrFaMo2009} presented a generalization based 
on the colorful \Caratheodory theorem.\footnote{Actually, Arocha~\etal 
present an even stronger result (the ``very colorful Kirchberger
theorem''~\cite[Theorem~3]{ArochaBaBrFaMo2009}) using a generalization 
of the colorful \Caratheodory theorem. Here, we consider the weaker 
version that can be obtained from Theorem~\ref{thm:colcara}.}

\begin{theorem}[Colorful Kirchberger theorem~{\cite[special case of
Theorem~3]{ArochaBaBrFaMo2009}}]\label{thm:colkirchberger}
Let $C_1,\dots, C_n \subset \R^d$ be $n=(m-1)(d+1)+1$ pairwise disjoint 
color classes and let $\mc{T}_i = \set{T_{i,1},\dots,T_{i,m}}$ 
denote a Tverberg
  $m$-partition for $C_i$, where $i \in [n]$. Then, there exists a colorful
  choice $C$, $|C|=n$, such that the family of sets
\[
  \mc{T}_{C} = \left\{ C \cap \left( \bigcup_{i=1}^n T_{i,j} \right) 
  \midd j \in [m] \right\}
\]
  is a Tverberg $m$-partition for $C$.
\end{theorem}
\begin{prf}
  We lift each Tverberg partition to $\R^{n-1}$ as in 
  Lemma~\ref{lem:sarkaria}: for $i \in [n]$ and $j \in [m]$, we 
  denote with $\UP{T}_{i,j}$ the set
\[
  \UP{T}_{i,j} = \set{ \TwoRowVec{\pp}{1} \otimes \qq_j \midd 
  \pp \in T_{i,j}} \subset
  \R^{n-1}.
\]
By Lemma~\ref{lem:sarkaria} and since each set 
$\mc{T}_i$, $i \in [n]$, is a Tverberg
partition, the sets $\UP{C}_i = \bigcup_{j=1}^m \UP{T}_{i,j}$, $i \in [n]$,
embrace the origin. We color the points in $\UP{C}_i$ with color $i$. Since
there are $n$ color classes that embrace the origin in $n-1$ dimensions,
Theorem~\ref{thm:colcara} guarantees the existence of a colorful 
choice $\UP{C}$ that
embraces the origin. For $j \in[m]$, let $\UP{T}_{j}=\UP{C} \cap \left(
\bigcup_{i=1}^n \UP{T}_{i,j} \right)$ denote all points from a $j$th 
element in a Tverberg partition in $\UP{C}C$. Since 
$\UP{C}= \bigcup_{j=1}^m \UP{T}_{j}$ embraces
the origin, Lemma~\ref{lem:sarkaria} implies that the convex hulls of
the sets 
$T_j = \set{ \pp \in \bigcup_{i=1}^n P_i \midd \TwoRowVec{\pp}{1} 
\otimes \qq_j \in \UP{T}_j}$ have a nonempty intersection. 
Further, since for $j \in [m]$, the set
$\UP{T}_{j}$ is a subset of $\bigcup_{i=1}^n \UP{T}_{i,j}$, we have $T_j
\subset \left( \bigcup_{i=1}^n T_{i,j} \right)$. Moreover, since all points
that correspond to the Tverberg partition $\mc{T}_i$, $i \in [n]$, have
color $i$, exactly one of the sets $T_1,\dots,T_m$ contains a point 
from $C_i$. The colorful choice $C$ can be obtained by projecting 
$\UP{C}$ down to $\R^d$.
\end{prf}

We now give precise bounds on the quality and the running time of 
approximation algorithms obtained by
combining algorithms for $k$-colorful choices with the presented
reductions to \CCP. Unfortunately, the approximation
guarantee of Algorithm~\ref{alg:mainapx} is too weak to obtain a nontrivial
approximation algorithm for \Tverberg and therefore also for \Centerpoint. 
On the positive side, it leads to a nontrivial approximation algorithm for
\ColKirchberger.

In the following, let $\mc{A}$ be an algorithm that, when given $d+1$ 
color classes
$C_1,\dots,C_{d+1} \subset \R^d$, each embracing the origin and of size
$\Oh{d}$, and for each $C_i$ the coefficients of the convex combination 
of the origin, outputs a $\0$-embracing $k(d)$-colorful choice in 
$W(d)$ time, where $k, W: \N \rightarrow \N$ are arbitrary but fixed functions.

\begin{corollary}\label{cor:app:tverberg}
  Let $P \subset \R^d$ be a point set of size $n$ and let $\mc{A}$ 
  be as above.  Set
  \[
  \widetilde{m} = 
  \lt\lceil \frac{n}{(d+1)^2\big(k(n-1)-1\big)+d+1}\rt\rceil =
  \Om{\frac{n}{d^2k(n-1)}}.
  \]
  Then, a Tverberg $\widetilde{m}$-partition $\mc{T}$ of
  $P$ and a point $\pp \in \bigcap_{T \in \mc{T}}
  \conv(T)$ can be computed in total time $\Oh{(d^2 + m) n^2 + W(n-1)}$.
\end{corollary}
\begin{prf}
\newcommand{\kch}[1]{\widetilde{#1}}
Set $m = \lt\lceil n / (d+1)\rt\rceil$.
In the proof of Theorem~\ref{thm:tverberg},
we lift $m$ copies of $P$ with Lemma~\ref{lem:sarkaria} to $\R^{n-1}$. 
Lifting one
point needs $\Oh{dm} = \Oh{n}$ time and hence lifting all $m$ copies takes
$\Oh{mn^2}$ time. Then, each point $\pp_i \in \R^d$ from $P$
corresponds to a color class $C_i =
\lt\{\up{\pp}_{i,j} \midd j \in [m]\rt\} \subset \R^{n-1}$
of size $m$ and a $\0$-embracing colorful choice of $C_1,\dots,C_n$
corresponds to the Tverberg partition $\mc{T}=\{T_1,\dots,T_m\}$ that 
we obtain
by assigning $\pp_i \in P$ to $T_j$ if $\up{\pp}_{i,j} \in C$.
By construction of the color classes in the proof of
Theorem~\ref{thm:tverberg}, the barycenter of $C_i$ is the origin, for 
$i \in [n]$.
Since we know then for each color class the coefficients of the convex
combination of
the origin, we can apply $\mc{A}$ to
obtain a $\0$-embracing $k(n-1)$-colorful choice $\kch{C} \subseteq
\bigcup_{i=1}^n C_i$ together with the coefficients of the convex 
combination of
the origin with the points in $\kch{C}$.
Let
$\kch{T}=\lt\{\kch{T}_1,\dots,\kch{T}_m\rt\}$ be a family of
subsets of $P$ that we construct as before by assigning $\pp_i$ to 
$\kch{T}_j$ if
$\up{\pp}_{i,j} \in \kch{C}$. Here, $\kch{T}$ is a multiset, i.e., we allow
$\kch{T}_i = \kch{T}_j$ for $i\neq
j$. Since $\kch{C}$ embraces the origin, Lemma~\ref{lem:sarkaria}
guarantees that the intersection $\bigcap_{i=1}^m \convv{\kch{T}_i}$ is
nonempty. Moreover, because we know the coefficients of the convex 
combination of the origin with the points in $\kch{C}$, we can 
compute in $\Oh{dn}$ time a
point $\pp^\star \in \bigcap_{i=1}^m \convv{\kch{T}_i}$ together with the
coefficients of the convex combination of $\pp^\star$ with the points in
$\kch{T}_i$ for $i \in [m]$, as described in the proof of 
Lemma~\ref{lem:sarkaria}.

Now, we construct a Tverberg partition for $P$ out of
$\kch{\mc{T}}$ by a greedy strategy that iteratively
removes sets from $\kch{\mc{T}}$. Let $\kch{T} \in
\kch{\mc{T}}$ be some set and remove it from $\kch{\mc{T}}$. 
Since we know the coefficients of the convex
combination of $\pp^\star$ with the points in $\kch{T}$,
Lemma~\ref{lem:constr_cara} can be applied to prune $\kch{T}$
to a $\pp^\star$-embracing set of size at most $d+1$ in 
$\Oh{d^3n + n^2}$ time. Then, for each point $\pp \in
\kch{T}$, we remove the at most $k(n-1)-1$ other sets from
$\kch{\mc{T}}$ that contain $\pp$. We continue with the next set
in $\kch{\mc{T}}$ that has not yet been removed until $\kch{\mc{T}} =
\emptyset$. Let $\mc{T}^\star \subseteq
\kch{\mc{T}}$ be the family of sets that we obtain by this process.
Clearly, $\mc{T}^\star$ is a Tverberg partition and 
because $\mc{T}^\star \subseteq \kch{\mc{T}}$, we have $\pp^\star \in
\bigcap_{\kch{T} \in \mc{T}^\star} \convv{\kch{T}}$.
Moreover, for each set
$\kch{T}_i \in \mc{T}^\star$, we remove at most $(d+1)(k(n-1)-1)$
other sets from $\kch{\mc{T}}$. Thus, the size of the Tverberg partition
$\mc{T}^\star$ is at least
\[
  \lt|\mc{T}^\star\rt|
  \geq \lt\lceil \frac{m}{(d+1)(k(n-1)-1)+1}\rt\rceil
  \geq \lt\lceil \frac{n}{(d+1)^2(k(n-1)-1)+d+1}\rt\rceil.
\]

Constructing the \CCP instance takes $\Oh{m n^2}$ time.
Using $\mc{A}$, we need $W(n-1)$ time to compute a $k(n-1)$-colorful choice
$\kch{C}$. Pruning every set of $\kch{\mc{T}}$ with 
Lemma~\ref{lem:constr_cara} to
at most $d+1$ points needs $\Oh{m (d^3 n + n^2)}=\Oh{(d^2 + m) n^2}$ time.
Finally, constructing
$\mc{T}^\star$ out of $\kch{\mc{T}}$ takes $\Oh{n^2}$ time with the naive
algorithm. This results in the claimed running time of
$\Oh{(d^2 + m) n^2 + W(n-1)}$.
\end{prf}

Furthermore, we can use $\mc{A}$ to approximate \ColKirchberger.

\begin{corollary}\label{cor:app:colkirchberger}
Let $\mc{A}$ be as above and let $C_1,\dots, C_n \subset \R^d$ be
$n=(m-1)(d+1)+1$ pairwise disjoint color
classes that are each of size $n$. Furthermore, for $i \in [n]$, let
$\mc{T}_i = \left\{T_{i,1},\dots,T_{i,m} \right\}$ denote a Tverberg
$m$-partition for $C_i$.
Then, given for each Tverberg partition $\mc{T}_i$, $i \in [n]$, a 
point $\pp_i
\in \bigcap_{j=1}^m \convv{T_{i,j}}$, and for all $i \in [n]$ and 
$j \in [m]$, the coefficients of the
convex combination of $\pp_i$ with the points in $T_{i,j}$,
we can compute in $\Oh{n^3 + W(n-1)}$ time a
$k(n-1)$-colorful choice $C\subseteq \bigcup_{i=1}^n C_i$ such that
\[
  \mc{T}_{C} = \left\{ C \cap \left( \bigcup_{i=1}^n T_{i,j} \right) \midd j \in
  [m] \right\}
\]
is a Tverberg $m$-partition for $C$.
\end{corollary}
\begin{prf}
  In the proof of Theorem~\ref{thm:colkirchberger},
  we lift the points $\bigcup_{i=1}^n C_i$ to $\R^{n-1}$ such that the
  set of points $\UP{C}_i$ that corresponds to the color class $C_i$
  still
  embraces the origin, where $i \in [n]$. Moreover, if $\UP{C}' \subseteq
  \bigcup_{i=1}^n \UP{C}_i$ is a $\0$-embracing colorful choice of the 
  lifted points, then there is a corresponding colorful choice $C'$ 
  with respect to $C_1,\dots,C_n$ such that
\[
  \mc{T}_{C'} = \left\{ C' \cap \left( \bigcup_{i=1}^n T_{i,j} \right)
  \midd j \in [m] \right\}
\]
is a Tverberg $m$-partition for $C'$.
Similarly, a $\0$-embracing $k(n-1)$-colorful choice $\UP{C}$ of the
lifted color classes corresponds to a $k(n-1)$-colorful choice $C$ with 
respect to $C_1,\dots,C_n$ such that
\[
  \mc{T}_{C} = \left\{ C \cap \left( \bigcup_{i=1}^n T_{i,j} \right)
  \midd j \in [m] \right\}
\]
is a Tverberg $m$-partition for $C$.

Computing the tensor product $\TwoRowVec{\pp}{1} \otimes \qq$, where 
$\pp \in \R^{d}$ and $\qq \in
\R^{m-1}$, needs $\Oh{dm} = \Oh{n}$ time and hence
lifting the point sets $C_1,\dots,C_n \subset \R^d$ to $\R^{n-1}$ with
Lemma~\ref{lem:sarkaria} needs $\Oh{n^3}$ time in total.
Since we know for each Tverberg partition $\mc{T}_i$, $i \in [n]$, a point
$\pp_i \in \bigcap_{j=1}^m \convv{T_{i,j}}$ together with the
coefficients of the convex combination of $\pp_i$ with the points in 
$T_{i,j}$ for $j \in [m]$, we can compute in $\Oh{n}$ time the 
coefficients of the convex combination of the origin with the points 
in $\UP{C}_i$ as described in the proof of Lemma~\ref{lem:sarkaria}.
Then, $\mc{A}$ can be applied to compute a $\0$-embracing $k(n-1)$-colorful
choice $\UP{C}$ with respect to the lifted point sets in
$W(n-1)$ time. Finally, constructing $C$ and $\mc{T}_C$ out of
$\UP{C}$ needs
$\Oh{n}$ time. Hence, the total time needed is $\Oh{n^3 + W(n-1)}$.
\end{prf}

Now, given $d+1$ color classes $C_1,\dots,C_{d+1} \subset \R^d$ that 
embrace the origin, we can compute with Algorithm~\ref{alg:mainapx} an 
$\lt\lceil \eps d\rt\rceil$-colorful choice that embraces
the origin in polynomial time. Combining this
with Corollary~\ref{cor:app:tverberg},
we obtain an algorithm that computes Tverberg partitions of size 
$\Oh{1}$ in polynomial time, a trivial result.
However, combining Algorithm~\ref{alg:mainapx} with
Corollary~\ref{cor:app:colkirchberger}, we do obtain a nontrivial 
approximation algorithm for \ColKirchberger: given $n =
(m-1)(d+1)+1$ color classes $C_1,\dots,C_n$, each of size $n$, and for each
color class a Tverberg $m$-partition 
$\mc{T}_i = \left\{T_{i,1},\dots,T_{i,m} \right\}$ together with a
point $\pp_i \in \bigcap_{j=1}^m \convv{T_{i,j}}$ and the coefficients 
of the convex combination of $\pp_i$ with the points in $T_{i,j}$, for 
all $j \in [m]$, we can compute
in $n^{\Oh{\eps^{-1} \ln \eps^{-1}}}$ time
an $\lceil \eps n \rceil$-colorful choice $C$ such that
\[
  \mc{T}_{C} = \left\{ C \cap \left( \bigcup_{i=1}^n T_{i,j} \right)
  \midd j \in [m] \right\}
\]
is a Tverberg $m$-partition for $C$, where $\eps > 0$ is arbitrary but 
fixed.

\section{Exact Algorithms}
\label{sec:exact}
In contrast to the previous sections, we now focus on computing an
exact solution for the convex version of \CCP. Let $C_1,\dots,C_{d+1} 
\subset\Q^d$ be $d+1$ sets that each embrace the origin, and assume all 
are of size at most $d+1$. The naive algorithm checks for all 
$\Oh{d^{d+1}}$ possible colorful choices whether they embrace the origin. 
This can be further improved by using the following result by \Barany.

\begin{theorem}[{\cite[Theorem~2.3]{Barany1982}}]\label{thm:chooseone}
  Let $C_1,\dots,C_{d} \subset \R^d$ be $d$ sets that all embrace the 
  origin and let $\cc \in \R^d$ be a point.
  Then, there exist $d$ points $\cc_1 \in C_1,\dots,\cc_d \in C_d$ such 
  that the
  set $\left\{ \cc, \cc_1,\dots,\cc_d \right\}$ embraces the origin. \qed
\end{theorem}

In particular, Theorem~\ref{thm:chooseone} implies that every point 
$\cc \in \bigcup_{i=1}^{d+1} C_i$ participates in some $\0$-embracing 
colorful choice and hence we can fix a point from one color class and 
check only all $\Oh{d^d}$ possibilities of extending it to a colorful 
choice.

We now consider two related settings that allow for further improvement. 
We begin with the simple case in which
each color class consists of only two points 
(Section~\ref{sec:exact:simple}). Then,
basic linear algebra suffices to compute a $\0$-embracing colorful 
choice in polynomial-time. In Section~\ref{sec:exact:many}, we show that
many color classes help. Using an approach similar to the algorithm by 
Miller and Sheehy for approximating Tverberg 
partitions~\cite{MillerSh2010}, we present a quasi-polynomial time 
algorithm that computes a $\0$-embracing colorful choice
when given $\Th{d^2 \log d}$ color classes instead of only $d+1$.

\subsection{A Simple Special Case}\label{sec:exact:simple}

In the following, we assume that $|C_1| = \dots = |C_{d+1}| = 2$ and 
let $\cc_{i,1}, \cc_{i,2}$ denote the two points in $C_i$, for 
$i \in [d+1]$. Clearly, for all $i \in [d+1]$, the point $-\cc_{i,1}$ 
must be contained in the positive span of $\cc_{i,2}$. Furthermore, 
we assume without loss of generality that all points are different 
from the origin, as otherwise computing a $\0$-embracing colorful
choice is trivial. Then, the set $\lt\{ \cc_{i,1} \midd i \in [d+1] \rt\}$ 
is linearly dependent and hence there
exist coefficients $\phi_1,\dots,\phi_{d+1} \in \R$, not all $0$, such that
$\0 = \sum_{i=1}^{d+1} \phi_i \cc_{i,1}$. Now, since $-\cc_{i,1} \in
\poss{\cc_{i,2}}$ for all $i \in [d+1]$, the set $C = \lt\{
\cc_{i,1} \midd i \in [d+1],\, \phi_i \geq 0\rt\} \cup \lt\{
\cc_{i,2} \midd i \in [d+1],\, \phi_i < 0\rt\}$ embraces the origin,
and it is a colorful choice.
Since the computation of the coefficients of the linear dependency can
be carried out in $\Oh{d^3}$ time with Gaussian elimination, 
finding $C$ takes
$\Oh{d^3}$ time in total. The following theorem is now immediate.

\begin{theorem}\label{thm:2pointsperclass}
  Let $C_1, \dots, C_{d+1}\subset \R^d$ be $d+1$ pairs of points that all
  embrace the origin. Then, a $\0$-embracing colorful choice can be 
  computed in $\Oh{d^3}$ time.
\end{theorem}

\subsection{Many Colors}\label{sec:exact:many}

In the following, we assume that we are given $\Th{d^2 \log d}$ instead of
only $d+1$ color classes that all embrace the origin. The algorithm
repeatedly combines $k$-colorful choices to
one $\0$-embracing $\lceil k/2\rceil$-colorful choice until a $\0$-embracing
$1$-colorful choice is obtained. This approach is similar to the
Miller-Sheehy approximation algorithm for Tverberg
partitions~\cite{MillerSh2010}, and it leads to an algorithm with total 
running time $d^{\Oh{\log d }}$.

\begin{lemma}\label{lem:combine}
  Let $C'_1,\dots,C'_{d+1}\subset\R^d$ be $\0$-embracing $k$-colorful 
  choices of size $\Oh{d}$ such that each color appears in a unique 
  $k$-colorful choice. Then, given the coefficients of the convex 
  combination of the origin for each set $C'_i$, $i \in [d+1]$,
  a $\0$-embracing $\lceil k/2\rceil$-colorful choice $C'$ can be 
  computed in $\Oh{d^5}$ time.
\end{lemma}
\begin{prf}
First, we prune each $k$-colorful choice $C'_i$, $i \in[d+1]$, with
Lemma~\ref{lem:findminembr} and then partition it into two
sets $C'_{i,1}, C'_{i,2}$ by distributing the points from each color
equally among both sets. Then, we apply the algorithm from
Lemma~\ref{lem:representatives} to obtain two representative
points $\rr_{i,1}, \rr_{i,2}$ and set $R_i =
\{\rr_{i,1},\rr_{i,2}\}$. Since the sets $R_1,\dots,R_{d+1}$ each 
embrace the origin and consist of only two points, we can compute 
a $1$-colorful choice $R$ with respect to $R_1,\dots,R_{d+1}$ with 
the algorithm from Theorem~\ref{thm:2pointsperclass}.
Now, consider the set $C' = \lt\{ C'_{i,j} \midd \rr_{i,j} \in R \rt\}$. 
Since $R$ is $\0$-embracing, so is $C'$. Moreover, because a color 
$j$ appears only in one of the $k$-colorful choices, say $C'_i$, 
and since each set of the partition
$C'_{i,1}, C'_{i,2}$ contains at most $\lt\lceil k/2 \rt\rceil$ points
with color $j$, the set $C'$ is a $\lt\lceil k/2\rt\rceil$-colorful choice.

Pruning each $k$-colorful choice with Lemma~\ref{lem:findminembr} and 
then computing the two representative points per partition takes 
$\Oh{d^5}$ time in total. This dominates the time needed for the 
computation of $R$ and thus, we can compute $C'$ in $\Oh{d^5}$ time.
\end{prf}

Note that Lemma~\ref{lem:combine} actually implies a second algorithm 
to compute $\lt\lceil (d+1)/2 \rt\rceil$-colorful choices that 
embrace the origin: let $C_1,\dots,C_{d+1} \subset \R^d$ be 
$\0$-embracing color classes and assume the sets have size $d+1$. Set
$C'_i = C_i$ in Lemma~\ref{lem:combine}, for $i \in [d+1]$. 
Then, $C'_i$ is trivially a $(d+1)$-colorful choice and hence 
the set $C'$ is a $\lt\lceil (d+1)/2 \rt\rceil$-colorful choice.

Now, we apply Lemma~\ref{lem:combine} repeatedly until we obtain a 
$1$-colorful choice as follows. Let $C_1,\dots,C_n \subset \Q^d$ be 
$n=\Th{d^2\log d}$ color classes such that
$C_i$ is $\0$-embracing and has size $d + 1$, for $i \in [n]$.
We create an array $A$ of size $m=\Th{\log d}$ that
initially contains all $n$ color classes in $A[0]$. Set $c_0 =
d+1$ and for $i \in [k]$, set $c_i = \lceil c_{i-1}/2\rceil$.
Throughout the algorithm, we maintain the invariant that the $i$th 
cell contains only $\0$-embracing $c_i$-colorful choices and that 
each color appears in at most one set in all of $A$. Since 
$c_0 = d+1$, the invariant holds in the
beginning. We repeatedly improve $k$-colorful choices with
Lemma~\ref{lem:combine} as follows: let $i$ be the maximum index of 
a cell in $A$ that contains at least $d+1$ sets $C'_1,\dots,C'_{d+1}$ 
and remove them from $A[i]$. By our invariant, these sets are 
$\0$-embracing $c_i$-colorful choices. Applying 
Lemma~\ref{lem:combine}, we can combine $C'_1,\dots,C'_{d+1}$ to one
$c_{i+1}$-colorful choice $C'$ that embraces the origin. We prune it with
Lemma~\ref{lem:findminembr} and check whether it is a $1$-colorful 
choice. If so, we have found a solution. Otherwise, we add it to 
$A[i+1]$. Furthermore, we check for colors that appeared in the 
removed sets $C'_1,\dots,C'_{d+1}$ but
not in $C'$ and add the corresponding color classes back to $A[0]$. The
invariant is maintained since these colors only appeared in the 
removed sets. See Algorithm~\ref{alg:many} for a detailed 
description of the algorithm.

\begin{alg}
  \KwIn{color classes $C_1,\ldots,C_{n} \subset \R^d$ and for 
  each set $C_i$, the coefficients of the convex combination of $\0$, 
  where $n=\Th{d^2 \log d}$}
  $A \gets $ Array of size $m=\Th{\log d}$\;
 Prune $C_1, \dots, C_n$ with Lemma~\ref{lem:findminembr}\;
  $A[0] \gets \lt\{C_1,\dots,C_n\rt\}$\;
  \While{no $\0$-embracing colorful choice was found}{%
    $i \gets $ maximum index with $|A[i]| \geq d+1$\;
    Remove $d+1$ sets $C'_1,\dots,C'_{d+1}$ from $A[i]$\;
    $C' \gets $ combine $C'_1,\dots,C'_{d+1}$ with Lemma~\ref{lem:combine}\;
    Prune $C'$ with Lemma~\ref{lem:findminembr}\;
    \If{$C'$ is a colorful choice}{%
    \Return{$C'$}\;
    }
      Add $C'$ to $A[i+1]$\;
    Add all color classes $C_i$ with 
    $C_i \cap \lt(\bigcup_{i=1}^{d+1} C'_i\rt) \neq
    \emptyset$ and $C_i \cap C' = \emptyset$ to $A[0]$\;
  }
\caption[Solving \CCP with many color classes exactly.]{Exact 
algorithm for many color classes.}
\label{alg:many}
\end{alg}

We conclude this section by proving the correctness of 
Algorithm~\ref{alg:many} and analyzing its running time.

\begin{theorem}\label{thm:exact}
  Let $C_1,\ldots,C_{n} \subset \R^d$ be $n=\Th{d^2 \log d}$ sets such that
  $C_i$ embraces the origin and $|C_i| = \Oh{d}$, for $i \in[n]$.
  Then, given the coefficients of the convex combination of the 
  origin for each set $C_i$, $i \in [n]$, Algorithm~\ref{alg:many} 
  computes a $\0$-embracing colorful choice in $d^{\Oh{\log d }}$ time.
\end{theorem}
\begin{prf}
  Set $m = \lceil \log(d+1) \rceil+1$.
  We have already argued that the $i$th cell of the array $A$ contains 
  only $\0$-embracing $c_i$-colorful choices.
  First, we observe that progress is always possible, i.e., that 
  it is always possible to find a cell of $A$ that contains at least 
  $d+1$ sets: the array has $m = \Th{\log d}$ levels and within each 
  set in $A$,
  at most $d$ colors appear. Thus, for $d^2 m + 1 = \Th{d^2 \log d}$ colors,
  the pigeonhole principle guarantees a cell with at least $d+1$
  sets.

   We claim that a combination of $d+1$ sets in $A[m]$ results in 
   a $\0$-embracing colorful choice. Since $c_i \leq \frac{d+1}{2^i} + 2$,
  the sets in $A[m-1]$ are $\0$-embracing $3$-colorful choices, the 
  sets in $A[m]$ are $2$-colorful choices and the combination of 
  $d+1$ sets in $A[m]$ gives a $1$-colorful choice, as claimed.

  Let $T(i)$ denote the time to compute a set at level $i$.
  For this,  we have to compute
  $d+1$ sets in level $i-1$. Since one application of 
  Lemma~\ref{lem:combine} takes
  $\Oh{d^5}$ time, we have $T(i) = (d+1)T(i-1) + \Oh{d^5}$, for $i \geq 1$,
  and $T(0) = \Oh{1}$. This solves to $T(i) = d^{\Oh{i}}$. 
  At the end, each level $i \geq 1$ of $A$ contains at most $d+1$ sets, 
  so the total running time is 
  $\sum_{i = 1}^{m+1} (d+1)T(i) = \sum_{i = 1}^{m+1} d^{\Oh{i}} = 
  d^{\Oh{\log d}}$, as claimed.
\end{prf}

\section{The Complexity of a Related Problem}
\label{sec:ncp}
We can show that a related problem to \CCP that is motivated by 
\Barany's original proof~\cite{Barany1982}, the \emph{local search 
nearest colorful polytope problem} (\NCPl), is \PLS-complete.
Additionally, by adapting the \PLS-completeness proof, we prove that
finding a global optimum for \NCP (\NCPg) is \NP-hard. This answers a 
question by \Barany and Onn~\cite[p.\ 561]{BaranyOn1997}. We note 
that this question has been answered independently by Meunier and
Sarrabezolles~\cite[Theorem~2]{MeunierSa2014}.
In contrast to the previous sections, all algorithms in this section are
analyzed in the \WordRAM with logarithmic costs.
This models the number of steps on a Turing machine, as required 
by the definition of \PLS.

\subsection{The Complexity Class \PLS}

The complexity class \emph{polynomial local search}
(\PLS)~\cite{JohnsonPaYa1988,AartsLe2003,AartsMiKo2007}
captures search problems that can be solved by a
local-improvement algorithm. Each improvement step can be
carried out in polynomial time, but the total number of 
steps to a local optimum may be exponential.
The existence of a local optimum is guaranteed, as the
progress of the algorithm can be measured by a potential function
that strictly decreases with each improvement step.

More formally, a problem in \PLS is a relation $\mc{R}$ between a set of
\emph{problem instances} $\mc{I} \subseteq \{0,1\}^\star$ and a set of
\emph{candidate solutions} $\mc{S} \subseteq \{0,1\}^\star$ with 
the following properties: 

\begin{itemize}
\item The set $\mc{I}$ is polynomial-time verifiable. Furthermore,
there exists an algorithm that, given an instance $I \in \mc{I}$ and a
candidate solution $s \in \mc{S}$, decides in time $\poly(|I|)$ whether
$s$ is \emph{valid} for $I$.
In the following, we denote with $\mc{S}_I \subseteq \mc{S}$
the set of valid candidate solutions for a given instance $I$.
\item There exists a polynomial-time algorithm that on input $I \in
\mc{I}$ returns a valid candidate solution $s_I \in \mc{S}_I$.
We call $s_I$ the \emph{standard solution}.
\item There exists a polynomial-time algorithm that on input $I \in
\mc{I}$ and $s \in \mc{S}_I$ returns a set $N_{I,s} \subseteq
\mc{S}_I$ of valid candidate solutions for $I$. We call $N_{I,s}$ the
\emph{neighborhood} of $s$.
\item There exists a polynomial-time algorithm that on input $I \in
\mc{I}$ and $s \in \mc{S}_I$ returns a number $c_{I,s} \in \Q$. We
call $c_{I,s}$ the \emph{cost} of $s$.
\end{itemize}

We say a candidate solution $s \in \mc{S}$ is a \emph{local optimum} for 
an instance $I \in \mc{I}$ if (i) $s \in \mc{S}_I$; and 
(ii) for all $s' \in N_{I,s}$, we have
$c_{I,s} \leq c_{I,s'}$ (minimization problem) or $c_{I,s} \geq
c_{I,s'}$ (maximization problem).
The relation $\mc{R}$ then consists of all pairs $(I,s)$ such that $s$ is a
local optimum for $I$. This formulation implies a simple algorithm, 
the \emph{standard algorithm}: begin with the standard solution, and 
repeatedly call the neighborhood-algorithm to improve the current solution
until a local optimum is reached. Although each iteration 
takes polynomial time, the total number of iterations
may be exponential, the time needed to cycle through all the
exponentially many candidate solutions. 
There are straightforward examples where this happens.
Moreover, there are \PLS-problems for which it is
\PSPACE-complete to compute the local optimum found by the standard
algorithm~\cite[Lemma~15]{AartsLe2003}.

Each problem instance $I$ of a \PLS-problem can be seen as a
simple search problem on a directed 
graph $G_I=(V,E)$. The nodes of $G_I$ are the
valid candidate solutions for $I$, and there is a directed edge from 
$u \in \mc{S}_I$ to $v \in \mc{S}_I$ if $v \in N_{I,u}$ and 
$c_{I,v} < c_{I,u}$ (minimization problem) or $c_{I,v} > c_{I,u}$
(maximization problem). Then, the
set of local optima for $I$ is precisely the set of \emph{sinks} 
in $G_I$, i.e., the set of nodes with outdegree $0$. Because the
costs induce a topological order on the graph, at least one sink exists.

Since \PLS contains relations and not languages, a different
concept of reduction is necessary to define complete problems. 
We say a $\PLS$
problem $A$ is \emph{\PLS-reducible} (or just \emph{reducible}) to a 
$\PLS$ problem $B$
if there exist two polynomial-time computable functions 
$f_{A \mapsto B}$ and $f_{B \mapsto A}$ with the following properties. 
Let $\mc{I}_{A}$ denote the set of
instances of $A$ and let $\mc{S}_{A}$ denote the set of candidate 
solutions of
$A$. Define $\mc{I}_{B}$ and $\mc{S}_{B}$ similarly.
The function $f_{A \mapsto B}: \mc{I}_{A} \rightarrow \mc{I}_{B}$ maps
problem instances of $A$ to problem instances
of $B$. The function $f_{B \mapsto A}: \mc{I}_{A} \times \mc{S}_{B}
\rightarrow \mc{S}_A$ maps
candidate solutions of $B$ to
candidate solutions of $A$ such that 
if $s_B \in \mc{S}_{B}$ is a candidate solution of $B$ with 
$\lt(f_{A\mapsto B}(I_A), s_B\rt) \in B$,
then $\lt(I_A, f_{B \mapsto A}(I_A, s_B)\rt) \in A$.\footnote{Recall
that $A$ and $B$ are \emph{relations} between problem instances
and candidate solutions.}
The existence of these two functions implies that any
polynomial-time algorithm for $B$ yields a polynomial-time algorithm
for $A$.
We say a problem $A \in \PLS$ is \emph{\PLS-complete} if all problems 
in \PLS can be \PLS-reduced to $A$. The canonical \PLS-complete problem is
\prob{FLIP}~\cite[Theorem~1]{JohnsonPaYa1988}: given a Boolean circuit of
polynomial size with $n$ inputs and $m$ outputs, find an input-assignment 
such that the resulting output interpreted as a number in binary cannot 
be decreased by flipping one bit in the input. The set of \PLS-complete 
problems includes, among various local search variants and heuristics 
for \NP-complete problems,
the Lin-Kernighan heuristic for the traveling salesman
problem~\cite{Papadimitriou1992LKHeur},
computing stable configurations in Hopfield neural
networks~\cite[Corollary~5.12]{SchafferYa1991}, and computing
pure Nash equilibria in congestion
games~\cite[Theorem~3]{FabrikantpaChTa2004PureNash}.

\subsection{The Local Search Nearest Colorful Polytope Problem}
Let $C_1,\dots,C_{m} \subset \Q^d$ be $m$ color classes that do not 
necessarily embrace the origin. For a given set $C' \subset \Q^d$, 
let $\delta(C') = 
\min \left\{ \| \cc \|_1 \mid \cc \in \conv(C')  \right\}$ 
denote the minimum $\ell_1$-norm of a point in $\conv(C')$. In 
\NCPl, we want to find a colorful choice $C$ such that 
$\delta(C)$ cannot be decreased by swapping a single point
with another point of the same color. 
In the language of \PLS, \NCPl is defined as follows.
\begin{definition}[\NCPl]\hfill
  \begin{description}
    \item[Instances.] The set of problem instances $\mc{I}$ 
    consists of all tuples $(C_1,\dots,C_m)$, where $d \in\N$ 
    and for $i \in [m]$, we have $C_i \subset \Q^d$.
    \item[Candidate solutions.] The set of candidate solutions 
    consists of all
    sets $C \subset \Q^d$, where $d \in \N$. For a fixed instance $I =
    (C_1,\dots,C_m) \in \mc{I}$, we define the set of
    valid candidate solutions $\sol_I$ of $I$ to be the set of all colorful
    choices with respect to $C_1,\dots,C_m$.
    \item[Cost function.] Let $s \in \sol_I$ be a colorful
    choice.  Then, the cost $\cost_{I,s}$ of $s$ with respect to $I$ is
    defined as $\delta(s)$. We want to minimize the costs.
    \item[Neighborhood.] Let $I \in \mc{I}$ be an instance and let $s
    \in\mc{S}_I$ be a valid candidate solution. Then, the set of neighbors
    $\nbr_{I,s}$ of $s$ consists of all colorful choices that can be
      obtained by swapping one point with another point of the same color 
      in $s$.
  \end{description}
\end{definition}
We reduce the \PLS-complete problem \MTSATl~\cite{SchafferYa1991} 
to \NCPl. In \MTSATl, we are given a $2$-CNF formula, i.e., a Boolean 
formula in conjunctive normal form in which each clause consists
of at most $2$ literals, and for each clause a weight. The task 
is to find an assignment such that the weighted sum of unsatisfied clauses
cannot be decreased by flipping a single variable. More formally,
\MTSATl is defined as follows.

\newcommand{\clause}[1]{\UP{#1}}
\begin{definition}[\MTSATl]\hfill
  \begin{description}
    \item[Instances.] The set of instances $\mc{I}'$ consists of
    all tuples $I = (n,K_1,\dots,K_m)$ such that $n \in \N$ and 
    for $i \in[n]$, the tuple $K_i$ has the form $(w_i, T_i, F_i)$, 
    where $w_i \in \Z$ and $T_i, F_i \subseteq [n]$ with
    $|T_i \cup F_i| \leq 2$ for all $i \in [n]$. Then, we identify 
    with $K_i$ the clause 
    $\clause{K}_i = \lt(\bigvee_{j \in T_j} x_j\rt) \vee \lt(\bigvee_{j
    \in F_j} \overline{x}_j\rt)$ with weight $w_i$, and we identify 
    with $I$ the $2$-CNF formula 
    $\clause{K}_1 \wedge \dots \wedge \clause{K}_m$ with
    variables $x_1,\dots,x_n$.
    \item[Candidate solutions.] The set of candidate solutions $\mc{S}'$
    contains all tuples $A=(v_1,\dots,v_n)$,
    where $n \in \N$ and $v_i \in \{0,1\}$ for $i \in [n]$. 
    Given an instance $I
    \in \inst'$ in which $n$ variables $x_1,\dots,x_n$ appear, we define
    the set of valid candidate solutions
    $\sol'_{I}$ for $I$ as the set of all $n$-tuples from $\sol'$.
    We interpret the $i$th entry of a tuple $A \in \sol'_{I}$ as
    an assignment to $x_i$ and we denote it with $A(x_i)$.
    \item[Cost function.] Let $I \in \mc{I}'$ be an instance. 
    Then, we define the cost $c'_{I,s}$ of a valid candidate 
    solution $s \in \mc{S}'_I$ as the
    sum of the weights of all unsatisfied clauses. We want to minimize 
    the cost.
    \item[Neighborhood.] Let $I \in \mc{I}'$ be an
    instance and $s \in \mc{S}'_I$ a tuple of size $n$. Then, the
    set of neighbors $\nbr'_{I,s}$ of $s$ consists of all tuples that can 
    be obtained by replacing the $i$th entry $A(x_i)$ with $1-A(x_i)$, 
    where $i \in [n]$.
  \end{description}
\end{definition}

The following theorem is due to Sch\"affer and Yannakakis.
\begin{theorem}[{\cite[Corollary 5.12]{SchafferYa1991}}]
  \MTSATl is \PLS-complete.\qed
\end{theorem}

We continue with the reduction from \MTSATl to \NCPl.

\begin{theorem}
  \label{thm:NCPl}
  \NCPl is \PLS-complete.
\end{theorem}

\begin{prf}
  Let $I'=(n,K_1,\dots,K_d) \in \mc{I}'$ be an 
  instance of \MTSATl. We construct an instance 
  $I \in \mc{I}$ of \NCPl in which each colorful choice $C$ 
  encodes an assignment $A_C$ such that the cost
  $c_{I, C}$ of $C$ equals the cost $c'_{I', A_C}$.

  For each variable $x_i$, we introduce a color class $X_i = \{\xx_i,
  \overline{\xx}_i\} $ consisting
  of two points in $\Q^d$ that encode whether $x_i$ is set to $1$ or $0$.
  We assign the $j$th dimension to the $j$th clause and set
  \begin{equation*}
    \lt(\xx_i\rt)_j =
    \begin{cases}
      -n w_j, & \text{if $x_i = 1$ satisfies $\clause{K}_j$, and}\\
    w_j, & \text{otherwise},
    \end{cases}
  \end{equation*}
  where $j \in [d]$. Similarly, we set
  \begin{equation*}
\lt(\overline{\xx_i}\rt)_j =
    \begin{cases}
      -n w_j, & \text{if $x_i = 0$ satisfies $\clause{K}_j$, and}\\
    w_j, & \text{otherwise},
    \end{cases}
  \end{equation*}
  where $j \in[d]$.  Then, a colorful choice $C$ of $X_1, \dots, X_m$
  corresponds to the assignment $A_C \in\sol'_{I'}$  that
  sets $x_i$ to $1$ if $\xx_i\in C$ and otherwise to $0$. Conversely,
  an assignment $A \in\sol'_{I'}$ can be interpreted directly as a 
  colorful choice $C$ of $X_1, \dots, X_m$

In the following, we construct an instance of \NCPl such that the
convex hull of a colorful choice $C$ contains the origin if projected 
onto the dimensions corresponding to clauses that are satisfied by 
$A_C$ (and hence do not contribute to the cost of $C$). Moreover,
if projected onto the subspace corresponding to the unsatisfied clauses,
$\delta(C)$ equals the total weight of unsatisfied clauses which then 
defines completely the cost of $C$.

We introduce additional helper color classes to decrease the distance to the
origin in dimensions that correspond to satisfied clauses.
In particular, we have for each clause $\clause{K}_i$, $i \in [d]$, a color
class $H_i = \{\hh_i\}$ consisting of a single point, where
  \begin{equation*}
    (\hh_i)_j = \begin{cases}
      (d+1) \left((n+2) - \frac{d}{d+1}\right) w_i, & 
      \text{if $j = i$, and} \\
      w_j, & \text{ otherwise,}
      \end{cases}
  \end{equation*}
  where $j \in [d]$. The last helper color class $H_{d+1} =
  \{\hh_{d+1}\}$ again contains a single point, but now all
  coordinates are
  set to the clause weights, i.e.,
\[
(\hh_{d+1})_j = w_j, \text{~for $j \in [d]$.}
\]
See Figure~\ref{fig:reduction} for an example.
\begin{figure*}[htbp]
  \begin{center}
    \includegraphics[width=\textwidth]{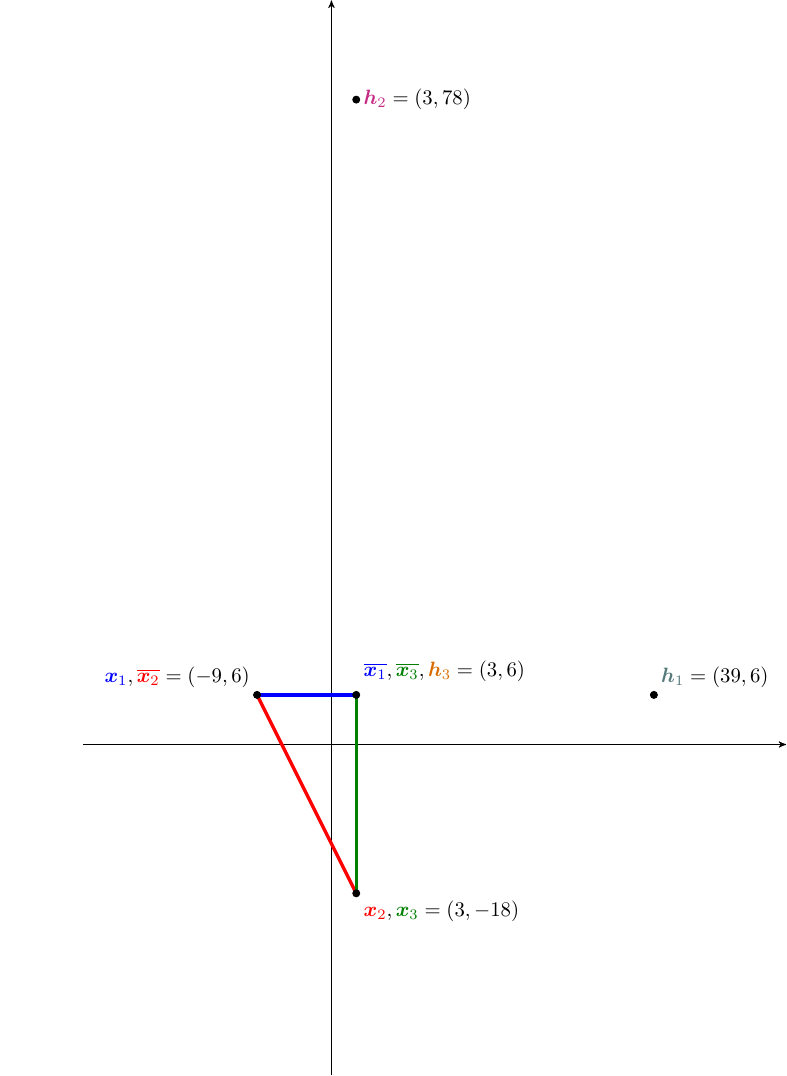}
  \end{center}
  \caption[Construction of the \NCPl instance that corresponds to
  a given \MTSATl instance.]{Construction of the point sets corresponding 
  to the \MTSATl instance 
  $\lt(x_1 \vee \overline{x_2}\rt) \wedge \lt(x_2 \vee
  x_3\rt)$ with weights $3$ and $6$, respectively.}
  \label{fig:reduction}
\end{figure*}

Let now $I = (X_1,\ldots,X_n, H_1, \ldots, H_{d+1}) \in \mc{I}$ denote the
constructed \NCPl instance. We continue with showing that the cost of a
colorful choice equals the cost of the corresponding assignment by 
proving the
following two inequalities.
\begin{enumerate}[label=(\roman{enumi})]
\item\label{pls:costs:lower}
for every colorful choice $C \in \mc{S}_I$, the cost are lower bounded by
the cost of the corresponding assignment:
\[
\cost_{I,C} \geq \cost'_{I',A_C}.
\]
\item\label{pls:costs:upper}
for every colorful choice $C \in \mc{S}_I$, the cost are upper bounded by
the cost of the corresponding assignment:
\[
\cost_{I,C} \leq \cost'_{I',A_C}.
\]
\end{enumerate}

Note that~\ref{pls:costs:lower} and~\ref{pls:costs:upper} directly 
imply that \NCPl is \PLS-complete. To see this, consider a local 
optimum $s^\star \in \mc{S}_I$ of the \NCPl instance $I$. By 
definition, the costs of all other colorful choices that can be 
obtained from $s^\star$ by swapping one point with another
of the same color are greater or equal to $\cost_{I,s^\star}$. Then, 
the total weight of unsatisfied clauses by the corresponding 
assignment $A_{s^\star} \in \mc{S}'_{I'}$ cannot be decreased by 
flipping a variable. Thus, $A_{s^\star}$ is a local minimum 
of the \MTSATl instance $I'$.

\ref{pls:costs:lower} Let $C \in \mc{S}_I$ be a colorful choice 
and assume some clause $\clause{K}_j$ is not satisfied by the 
corresponding assignment $A_C \in \mc{S}'_{I'}$. By
construction, the $j$th coordinate
of each point $\pp$ in $C$ is at least $w_j$. Thus, the $j$th coordinate
of every convex combination of the points in $C$ is at least $w_j$ and
hence $\cost_{I,C} \geq \cost_{I',A_C}$.

\ref{pls:costs:upper}
Let $C \in \mc{S}_I$ be a colorful choice. In the following, we construct a
convex combination of the points in $C$ that results in a point $\pp$ whose
$\ell_1$-norm is exactly the total weight of unsatisfied clauses in the
corresponding assignment $A_C \in \mc{S}'_{I'}$ and thus $\cost_{I,C} \leq
\cost_{I',A_C}$. For $k=0,1,2$, let $S_k$ denote the set of clauses that are
satisfied by exactly $k$ literals with respect to the assignment $A_C$. As a
first step towards constructing $\pp$, we show the existence of an 
intermediate
point in the convex hull of the helper classes.
\begin{lemma}\label{lem:h}
There is a point $\hh\in\conv(H_1,\ldots,H_{d+1})$ whose $j$th coordinate is
$(n+2) w_j$, if $j\in S_2$, and $w_j$, otherwise.
\end{lemma}
\begin{prf}
Take $\hh = \sum\limits_{i\in S_2} \frac{1}{d+1} \hh_i +
\left(1-\frac{|S_2|}{d+1}\right) \hh_{d+1}$. Then, for $j\in S_0\cup S_1$, 
we
have
\[
  (\hh)_j
  = \sum\limits_{i\in S_2} \frac{1}{d+1} (\hh_i)_j
      + \left(1-\frac{|S_2|}{d+1}\right) (\hh_{d+1})_j
    \stackrel{j\notin S_2}{=} \sum\limits_{i\in S_2} \frac{1}{d+1} w_j
      + \left(1-\frac{|S_2|}{d+1}\right) w_j 
  = w_j.
\]

And for $j \in S_2$, we have
\begin{align*}
    (\hh)_j & = \sum\limits_{i\in S_2} \frac{1}{d+1} (\hh_i)_j +
      \left(1-\frac{|S_2|}{d+1}\right) (\hh_{d+1})_j
    \\
      & = \frac{1}{d+1} (\hh_j)_j + \sum\limits_{i\in S_2\setminus\{j\}}
        \frac{1}{d+1} (\hh_i)_j +
        \left(1-\frac{|S_2|}{d+1}\right) (\hh_{d+1})_j
  \\
    & = \left((n+2) - \frac{d}{d+1}\right) w_j + \frac{d}{d+1}
      w_j = (n+2) w_j,
\end{align*}
as desired.
\end{prf}

Let now $\ve{l}_i$ be the point from $X_i$ in the colorful choice 
$C$ and consider the point
\[
\pp = \frac{1}{n+1} \lt(\sum_{i=1}^n  \ve{l}_i + \hh\rt),
\]
where
$\hh$ is the point from Lemma~\ref{lem:h}.
We show that $(\pp)_j = w_j$ if $j\in S_0$, and otherwise $(\pp)_j = 0$.
Let $j$ be a clause index from $S_0$. Since $A_C$ does not 
satisfy $\clause{K}_j$, the $j$th coordinate of the points 
$\ve{l}_1,\ldots,\ve{l}_n$ is $w_j$.
Also, $(\hh)_j = w_j$ by Lemma~\ref{lem:h}. Thus, 
$(\pp)_j = w_j$. Consider now
some clause index $j\in S_1$ and let $b \in [2]$ be the index of the point
$\ve{l}_b$ that corresponds to the single literal that satisfies 
$\clause{K}_j$.  Then, we have
\begin{multline*}
    (\pp)_j  = \sum_{i=1}^n \frac{1}{n+1} (\ve{l}_i)_j + \frac{1}{n+1}
      (\hh)_j\\
    = \frac{1}{n+1} (\ve{l}_b)_j + \sum_{i=1, i\neq b}^n
    \frac{1}{n+1} (\ve{l}_i)_j + \frac{1}{n+1} (\hh)_j
      =  \frac{-n}{n+1} w_j + \frac{n}{n+1} w_j = 0.
\end{multline*}
Finally, consider some clause index $j\in S_2$ and let $b_1, b_2$ 
be the indices of the two literals that satisfy $\clause{K}_j$. 
Then, we obtain
\begin{align*}
    (\pp)_j & = \sum\limits_{i=1}^n \frac{1}{n+1} (\ve{l}_i)_j + 
    \frac{1}{n+1} (\hh)_j
  \\
    & = \frac{1}{n+1} (\ve{l}_{b_1})_j +\frac{1}{n+1} (\ve{l}_{b_2})_j +
      \sum\limits_{i=1, i\notin \{b_1, b_2\}}^n \frac{1}{n+1} (\ve{l}_i)_j +
      \frac{1}{n+1} (\hh)_j
  \\
    & =  \frac{-2n}{n+1} w_j + \frac{n-2}{n+1} w_j + 
    \frac{n+2}{n+1} w_j = 0,
\end{align*}
and thus $\| \pp \|_1 = c_{I',A_C}$, as claimed.
\end{prf}

\subsection{The Global Search Nearest Colorful Polytope Problem}

In the global search variant \NCPg of the nearest colorful polytope
problem, we are looking for a colorful
choice $C$ such that $\delta(C)$ is minimum over all possible colorful
choices.
The proof of Theorem~\ref{thm:NCPl} can be adapted easily
to reduce \ThreeSat to \NCPg.

\begin{restatable}{theorem}{NCPgthm}\label{thm:NCPg}
  \NCPg is \NP-hard.
\end{restatable}
\begin{prf}
 Given a set of clauses $K_1,\dots,K_d$, we set the
  weight of each clause to $1$ and construct the same
  point sets as in the \PLS-reduction. Additionally, we introduce for each
  clause $K_j$ a new helper color class $H'_j = \{\hh'_j\}$, where
  \begin{align*}
    (\hh'_i)_j =
      \begin{cases}
        (d+1) \left((2n+3) - \frac{d}{d+1}\right), & \text{if $i=j$, and} \\
        1, & \text{otherwise.}
      \end{cases}
  \end{align*}
  Let now $C$ be a colorful choice and let $A_C$ be the
  corresponding assignment. As in the \PLS-reduction, for 
  $k = 0, \dots, 3$, let $S_k$ contain all clauses that are 
  satisfied by exactly $k$ literals in the assignment $A_C$. Then, 
  the following point $\hh$ is contained in
  the convex hull of the helper points:
  \begin{align*}
      \hh = \sum\limits_{i\in S_2} \frac{\hh_i}{d+1} +
          \sum\limits_{j\in S_3} \frac{\hh'_{j}}{d+1} +
      \left(1-\frac{|S_2|+|S_3|}{d+1}\right) \hh_{d+1}.
  \end{align*}
  As above, we see that $(\hh)_j = 1$, if $j \in S_0 \cup S_1$,
  $(\hh)_j = n+2$, if $j \in S_2$, and $(\hh)_j = 2n+3$, if $j \in S_3$.
  Indeed, for $j \in S_0 \cup S_1$, we have:
\begin{multline*}
  (\hh)_j
  = \sum\limits_{i\in S_2} \frac{1}{d+1} (\hh_i)_j
  + \sum\limits_{i\in S_3} \frac{1}{d+1} (\hh'_i)_j
      + \left(1-\frac{|S_2|+|S_3|}{d+1}\right) (\hh_{d+1})_j
  \\
    \stackrel{j\notin S_2 \cup S_3}{=} \sum\limits_{i\in S_2 \cup S_3} 
    \frac{1}{d+1} 
      + \left(1-\frac{|S_2| + |S_3|}{d+1}\right)  
  = 1.
\end{multline*}
 For $j \in S_2$, we have
\begin{align*}
    (\hh)_j & = \sum\limits_{i\in S_2} \frac{1}{d+1} (\hh_i)_j 
  + \sum\limits_{i\in S_3} \frac{1}{d+1} (\hh'_i)_j
 +     \left(1-\frac{|S_2|+|S_3|}{d+1}\right) (\hh_{d+1})_j
  \\
    & = (\hh_j)_j  + 
    \sum\limits_{i\in S_2 \setminus j} \frac{1}{d+1}  
  + \sum\limits_{i\in S_3} \frac{1}{d+1} 
 +     \left(1-\frac{|S_2|+|S_3|}{d+1}\right) 
  \\
    & = \left((n+2) - \frac{d}{d+1}\right)  + \frac{d}{d+1}
      = n+2,
\end{align*}
 and for $j \in S_3$, 
\begin{align*}
    (\hh)_j & = \sum\limits_{i\in S_2} \frac{1}{d+1} (\hh_i)_j 
  + \sum\limits_{i\in S_3} \frac{1}{d+1} (\hh'_i)_j
 +     \left(1-\frac{|S_2|+|S_3|}{d+1}\right) (\hh_{d+1})_j
  \\
    & = (\hh'_j)_j  + 
    \sum\limits_{i\in S_2 } \frac{1}{d+1}  
  + \sum\limits_{i\in S_3 \setminus j} \frac{1}{d+1} 
 +     \left(1-\frac{|S_2|+|S_3|}{d+1}\right) 
  \\
    & = \left((2n+3) - \frac{d}{d+1}\right)  + \frac{d}{d+1}
      = 2n+3.
\end{align*}

  As before, the convex combination
    $\pp = \sum_{i=1}^n \frac{1}{n+1} \ve{l}_i + \frac{1}{n+1} \hh$
  results in a point in the convex hull of $C$ whose distance to the 
  origin is the number of unsatisfied clauses, where $\ve{l}_i$ denotes 
  the point from $X_i$ in $C$. Indeed, if $\clause{K}_j$ is not satisfied,
  then all $j$-components in the sum are $1$, and $(\pp)_j = 1$.
   If $j \in S_1$, then, as discussed above
   \[
   (\pp)_j = \frac{-n}{n+1} + \frac{n-1}{n+1} + \frac{1}{n+1} = 0.
   \]
   If $j \in S_2$, then
    \[
   (\pp)_j = \frac{-2n}{n+1} +\frac{n-2}{n+1}  + \frac{n+2}{n+1}= 0,
   \]
   and if $j \in S_3$, then
     \[
   (\pp)_j = \frac{-3n}{n+1} +\frac{n-3}{n+1}  + \frac{2n+3}{n+1}= 0.
   \]

  Together with \ref{pls:costs:lower} from the proof of
  Theorem~\ref{thm:NCPl}, \ThreeSat can be decided by knowing a 
  global optimum $C^\star$ to the \NCP problem: if 
  $\delta(C^\star) = 0$, $A_{C^\star}$ is a
  satisfying assignment. If not, there exists no satisfying 
  assignment at all.
\end{prf}

As mentioned above, we can adapt the proof of
Theorem~\ref{thm:NCPg} to answer a question by \Barany and
Onn~\cite{BaranyOn1997}. 

\begin{corollary}
  Let $C_1,\dots,C_m\subset\Q^d$ be an input for \NCPg.
  Then, \NCPg remains \NP-hard even if $m=d+1$.
\end{corollary}
\begin{prf}
  Let $F$ be a \ThreeSat formula with $d$ clauses and $n$ variables.
  As in the proof of Theorem~\ref{thm:NCPg}, we construct 
  $n+2d+1 =: d'+1$ point
  sets in $\Q^d$ such that there is a colorful choice that embraces 
  the origin
  if and only if $F$ is satisfiable. Since $d' > d$, we 
  can lift the point sets to $\Q^{d'}$ by appending $0$-coordinates. 
  Then, we have $d'+1$ point sets such that there is a colorful choice 
  that embraces the origin if and only if
  $F$ is satisfiable.
\end{prf}

\section{Conclusion}
We conclude with several interesting open problems.
\begin{itemize}
  \item The algorithm in Theorem~\ref{thm:bapx} computes in polynomial 
  time a $\0$-embracing $\lceil \eps d \rceil$-colorful choice for 
  any fixed $\eps > 0$. A more careful analysis shows that the 
  algorithm needs only $c_\eps$ color classes, where $c_\eps>0$ is 
  a constant depending only on $\eps$. Hence, the algorithm does not
    use its complete input. Can this be used to further improve the
    approximation guarantee?
  \item Is it possible to compute a $\0$-embracing $o(d)$-colorful choice 
  in polynomial time and in particular, is it possible to compute a 
  $\0$-embracing $O(1)$-colorful choice in polynomial time?
  \item On the other hand, can it be shown that computing a 
  $\0$-embracing $O(1)$-colorful
    choice is as hard as computing a $\0$-embracing $1$-colorful choice?
  \item In Section~\ref{sec:exact}, we show that many color classes help
    to find a $\0$-embracing $1$-colorful choice. Can a $\0$-embracing
    $1$-colorful choice be computed in polynomial time if we have
    $\text{poly}(d)$ color classes?
\end{itemize}

\paragraph{Acknowledgements.}
We would like to thank Fr\'ed\'eric Meunier and Pauline Sarrabezolles
for interesting discussions on the colorful \Caratheodory problem
and for hosting us during multiple research stays at the
\'Ecole Nationale des Ponts et Chauss\'ees.
Furthermore, we would like to thank the anonymous reviewers for their 
detailed reading of our paper and for their
helpful and encouraging comments on previous versions.

\bibliographystyle{abbrv}      
\bibliography{refs}   
\end{document}